\documentclass[oneside,reqno,english]{amsart}
\usepackage[T1]{fontenc}
\usepackage[latin9]{inputenc}
\usepackage{amstext}
\usepackage{amsthm}
\usepackage{amssymb}
\usepackage{stmaryrd}

\makeatletter
\numberwithin{equation}{section}
\numberwithin{figure}{section}
\theoremstyle{plain}
\newtheorem{thm}{\protect\theoremname}
\theoremstyle{plain}
\newtheorem{lem}[thm]{\protect\lemmaname}
\theoremstyle{plain}
\newtheorem{cor}[thm]{\protect\corollaryname}
\theoremstyle{plain}
\newtheorem{prop}[thm]{\protect\propositionname}

\numberwithin{thm}{section}

\usepackage{babel}
\providecommand{\corollaryname}{Corollary}
\providecommand{\lemmaname}{Lemma}
\providecommand{\propositionname}{Proposition}
\providecommand{\theoremname}{Theorem}

\makeatother

\usepackage{babel}
\providecommand{\corollaryname}{Corollary}
\providecommand{\lemmaname}{Lemma}
\providecommand{\propositionname}{Proposition}
\providecommand{\theoremname}{Theorem}

\begin{document}
\title[Polaron Hydrogenic Atom in a Strong Magnetic Field]{Ground State of the Polaron Hydrogenic Atom in a Strong Magnetic
Field }
\address{\hspace{-4ex}$\,$Institut for Matematik, Aarhus Universitet, Ny
Munkegade 8000 Aarhus C, Danmark, and Institut for Matematiske Fag,
Aalborg Universitet, Skjernvej 4A 9220 Aalborg Øst, Danmark}
\address{\hspace{-4ex}ghanta@math.au.dk}
\thanks{\hspace{-3.5ex}The paper is based on work supported by NSF DMS-1600560.}
\author{Rohan Ghanta }
\begin{abstract}
The ground-state electron density of a polaron bound to a Co-ulomb
potential in a homogeneous magnetic field--the transverse coordinates
integrated out--converges pointwise and weakly in the strong magnetic
field limit to the square of a hyperbolic secant function.
\end{abstract}

\maketitle

\section{Introduction}

\noindent A non-relativistic Hydrogen atom in a strong magnetic field
interacting with the quantized longitudinal optical modes of an ionic
crystal is considered within the framework of Fröhlich's 1950 polaron
model \cite{Fr=00003D0000F6hlich}. Starting with Platzman's variational
treatment in 1962 the polaron Hydrogenic atom has been of interest
for describing an electron bound to a donor impurity in a semiconductor
\cite{Platzman}. Its first rigorous examination however came much
later in 1988 from Löwen who disproved several longstanding claims
about a self-trapping transition \cite{L=0000F6wen}.

A study of the polaron Hydrogenic atom in strong magnetic fields was
initiated by Larsen in 1968 for interpreting cyclotron resonance measurements
in InSb \cite{Larsen}. The model has since been considered in formal
analogy to the Hydrogen atom in a magnetic field though the latter
was understood rigorously again much later in 1981 by Avron et al.
who proved several properties including the non-degeneracy of the
ground state \cite{Avron-Herbst-Simon}. Whether or not these Hydrogenic
properties indeed persist when a coupling to a quantized field is
turned on remains to be seen.

In any case polarons are the simplest Quantum Field Theory models,
yet their most basic features such as the effective mass, ground-state
energy and wave function cannot be evaluated explicitly. And while
several successful theories have been proposed over the years to approximate
the energy and effective mass of various polarons, they are built
entirely on unjustified, even questionable, Ansätze for the wave function.
The paper provides for the first time an explicit description of the
ground-state wave function of a polaron in an asymptotic regime.

For the polaron Hydrogenic atom in a homogeneous magnetic field its
ground-state electron density in the magnetic-field direction is shown
to converge in the strong field limit to the square of a hyperbolic
secant function: a sharp contrast to the paradigmatic Gaussian variational
wave functions \cite{Zorkani Belhissi Kartheuser},\cite{Shi Devreese}
\& ref. therein. The explicit limiting function is realized as a density
of the minimizer of a one-dimensional problem with a delta-function
potential describing the second leading-order term of the ground-state
energy cf. \cite{Baumgartner-Solovej-Yngvason},\cite{Frank-Geisinger},\cite{Lieb-Solovej-Yngvason}.

\section{Model and Main Result}

\noindent The Fröhlich model is defined by the Hamiltonian 
\begin{equation}
\mathbb{H}(B):=H_{B}-\partial_{3}^{2}-\beta\left|x\right|^{-1}+\mathcal{N}+\frac{\sqrt{\alpha}}{2\pi}\int_{\mathbb{R}^{3}}\left(\frac{a_{k}e^{ik\cdot x}}{\left|k\right|}+\frac{a_{k}^{\dagger}e^{-ik\cdot x}}{\left|k\right|}\right)dk\label{eq:h1}
\end{equation}
acting on the Hilbert space $\mathcal{H}:=L^{2}(\mathbb{R}^{3})\otimes\mathcal{F}$
where $\mathcal{F}:=\oplus_{n\geq0}\varotimes_{s}^{n}L^{2}(\mathbb{R}^{3})$
is a symmetric phonon Fock space over $L^{2}(\mathbb{R}^{3})$. The
creation and annihilation operators for a phonon mode $a_{k}^{\dagger}$
and $a_{k}$ act on $\mathcal{F}$ and satisfy $[a_{k},a_{k^{\prime}}^{\dagger}]=\delta(k-k^{\prime})$.
The energy of the phonon field is described by the operator $\mathcal{N}=\int_{\mathbb{R}^{3}}a_{k}^{\dagger}a_{k}\,dk.$
The kinetic energy of the electron is described by the operator $H_{B}-\partial_{3}^{2}$
acting on $L^{2}(\mathbb{R}^{3})$, where $H_{B}=\sum_{j=1,2}\left(-i\partial_{j}+A_{j}(\mathbf{x})\right)^{2}$
is the two-dimensional Landau Hamiltonian with the magnetic vector
potential $A\left(x_{1},x_{2},x_{3}\right)=B/2\left(-x_{2},x_{1},0\right)$
corresponding to a homogeneous magnetic field of strength $B\geq0$
in the $x_{3}$-direction; the transverse coordinates are denoted
by $x_{\perp}=(x_{1},x_{2})$. Furthermore $\inf\text{spec}\,H_{B}=B$.
The parameters $\alpha\geq0$, $\beta>0$ denote the strengths of
the Coulombic electron-phonon interaction and the localizing Coulomb
potential; the coupling function $1/\left|k\right|$ is proportional
to the square-root of the Fourier transform of the Coulomb interaction.
The ground-state energy is 
\begin{equation}
E_{0}(B):=\inf\left\{ \left\langle \Psi,\,\mathbb{H}(B)\Psi\right\rangle :\ \|\Psi\|=1,\,\Psi\in H_{A}^{1}\left(\mathbb{R}^{3}\right)\otimes\text{dom}\left(\sqrt{\mathcal{N}}\right)\right\} ,\label{eq:2-1}
\end{equation}
where $H_{A}^{1}\left(\mathbb{R}^{3}\right)$ is a magnetic Sobolev
space of order one. A ground state exists since $-i\nabla-\beta|x|^{-1}$
has a negative energy bound state in $L^{2}\left(\mathbb{R}^{3}\right)$
\cite{Griesemer-Lieb-Loss}.

Unlike previous treatments here the arguments remain valid for all
values of the parameters $\alpha\geq0$, $\beta>0$. First the large
$B$ asymptotics of the ground-state energy is derived; the main result
is given as Theorem \ref{thm:Wavefunction} below. Since the pioneering
work of Larsen the model has been considered only in the perturbative
regime $\alpha\ll\beta$, and the ground-state energy $E_{0}(B)$
has been approximated as the Hydrogenic energy 
\[
E_{_{H}}(B):=\text{inf spec }H_{B}-\partial_{3}^{2}-\beta\left|x\right|^{-1}
\]
with a supposedly small correction from the electron-phonon interaction.
The large $B$ asymptotics of the Hydrogenic energy was derived rigorously
in 1981 by Avron et al. \cite{Avron-Herbst-Simon} using ideas from
\cite{Blankenbecler-Goldberger-Simon},\cite{Simon}:
\begin{align}
E_{_{H}}(B)=B & -\frac{\beta^{2}}{4}\left(\ln B\right)^{2}+\beta^{2}\ln B\ln\ln B-\beta^{2}\left(-\gamma_{E}/2+\ln2\right)\ln B-\beta^{2}\left(\ln\ln B\right)^{2}\nonumber \\
 & +2\beta^{2}\left(-\gamma_{E}/2-1+\ln2\right)\ln\ln B+\mathcal{O}(1)\ \ \text{as}\ \ B\rightarrow\infty\label{eq:Avron-Herbst-Simon}
\end{align}
with $\gamma_{E}$ the Euler-Mascheroni constant, and the expansion
can be carried out to arbitrary order. The first three terms are understood
heuristically: For large $B$ the electron is tightly bound in the
transverse plane to the lowest Landau orbit while localized in the
magnetic-field direction by a one-dimensional effective Coulomb potential
that behaves to leading order like a delta well of strength $\beta\ln(B/(\ln B)^{2})$
\cite{Brummelhuis-Ruskai},\cite{Ruskai-Werner},\cite{Schiff Snyder};
see (\ref{eq:uppercut}) and Appendix B below. The electron motion
is effectively one-dimensional cf. \cite{Brummelhuis-Duclos}. The
second and third leading-order terms describe the dominant asymptotic
behavior of the ground-state energy of this one-dimensional electron
confined along the magnetic field. The pronounced anisotropy in the
system is reflected by the characteristic length scales of the electron
density in the transverse and the magnetic-field direction $1/\sqrt{B}$
and $1/\ln B$ respectively.

The above Hydrogenic heuristics still apply when a coupling to the
phonon field is introduced i.e. $\alpha>0$. For large $B$ the phonons
cannot follow the electron's rapid motion in the transverse plane
and so resign themselves to dressing its entire Landau orbit: not
only is the electron again localized in the magnetic-field direction
by the one-dimensional effective Coulomb potential, but also the electron-phonon
coupling function is now proportional to the square-root of the Fourier
transform of the same effective Coulomb interaction cf. \cite{Kochetov-Leschke-Smondyrev},\cite{Smondyrev}
and property (k) in \cite{Ruskai-Werner}; the system behaves as a
one-dimensional strongly coupled polaron to leading order with interaction
strength $\alpha\ln(B/(\ln B)^{2})$ confined along the magnetic field
by a delta well of strength $\beta\ln(B/(\ln B)^{2})$, i.e. in the
effective one-dimensional model the electron-phonon coupling is mediated
by the magnetic field. The analogous large $B$ asymptotics of the
polaron Hydrogenic energy is derived to second order: 
\begin{thm}
\label{thm:Energy Asymptotics}Let $E_{0}(B)$ be as defined in (\ref{eq:2-1})
above. Then 
\begin{equation}
E_{0}(B)=B+\mathfrak{e}_{_{0}}\left(\ln B\right)^{2}+\mathcal{O}\left(\left(\ln B\right)^{3/2}\right)\text{ }\text{as }B\rightarrow\infty\text{ with}\label{eq:thm1}
\end{equation}
\begin{align}
\mathfrak{e}_{_{0}}: & =\inf\left\{ \int_{\mathbb{R}}\left|\varphi^{\prime}\right|^{2}dx-\frac{\alpha}{2}\int_{\mathbb{R}}\left|\varphi\right|^{4}dx-\beta\left|\varphi(0)\right|^{2}:\int_{\mathbb{R}}\left|\varphi\right|^{2}dx=1\right\} \label{eq: p}\\
 & =-\frac{1}{48}\left(\alpha^{2}+6\alpha\beta+12\beta^{2}\right).\label{eq:p-2}
\end{align}
\end{thm}

\noindent Here the second leading-order term describes the dominant
asymptotic behavior of the ground-state energy of the effective one-dimensional
strongly coupled polaron confined along the magnetic field. It is
evaluated explicitly by minimizing a nonlinear functional. Furthermore
the cross term in (\ref{eq:p-2}) indicates for large $B$ the effect
of the electron-phonon interaction is not perturbative.

The large $B$ asymptotics for the polaron Hydrogenic energy is argued
differently from Avron et al.'s proof of (\ref{eq:Avron-Herbst-Simon})
and generalizes the result of Frank and Geisinger who proved (\ref{eq:thm1})-(\ref{eq:p-2})
when $\beta=0$ using upper and lower bounds to the ground-state energy
\cite{Frank-Geisinger}. Their upper bound is established with a trial
wave function. Their lower bound is established by showing that the
Hamiltonian when restricted to the lowest Landau level is bounded
from below in the sense of quadratic forms by an essentially one-dimensional
strong-coupling Hamiltonian; the strategy from \cite{Lieb Thomas}
is then used to arrive at the nonlinear minimization problem for the
second leading-order term along with lower order error terms cf. \cite{Ghanta}.

For proving Theorem \ref{thm:Energy Asymptotics} Frank et al.'s strategy
in \cite{Frank-Geisinger} applies mutatis mutandis. Here the argument
for the upper bound is simplified using the effective Coulomb potential,
given below in (\ref{eq:uppercut}), which plays an essential role
in Avron et al.'s proof of (\ref{eq:Avron-Herbst-Simon}) but is conspicuously
absent from \cite{Frank-Geisinger}. Here the argument for the lower
bound uses the bathtub principle \cite{Lieb-Loss} to extract a delta
function from the Coulomb potential in the lowest Landau level; just
like in \cite{Frank-Geisinger} the error terms are bounded by $\mathcal{O}((\ln B)^{3/2})$,
but--and this has been demonstrated recently in \cite{Frank-Seiringer}
for a strongly coupled polaron--the error terms in the lower bound
should be much smaller. Moreover the upper bound suggests the third
leading-order term again analogously to its Hydrogenic counterpart
in (\ref{eq:Avron-Herbst-Simon}) is $-4\mathfrak{e}_{_{0}}\ln B\ln\ln B$.

The expansion in (\ref{eq:thm1}) should be carried out to higher
order. I conjecture the first six leading-order terms behave analogously
to their Hydrogenic counterparts in (\ref{eq:Avron-Herbst-Simon}):
these should describe the leading asymptotics for the minimization
of a nonlinear functional arising naturally in the proof of the upper
bound and given below in (\ref{eq:classic}), which is a classical
approximation to the Fröhlich model of the strongly coupled one-dimensional
polaron confined along the magnetic field: in this classical approximation
the electron-phonon interaction is replaced with the one-dimensional
effective Coulomb self-interaction of the electron. The second leading-order
term above arises from this classical approximation by estimating
the one-dimensional effective Coulomb interaction as a delta interaction
of strength $\ln B$. Furthermore, in the seventh leading-order term
there should be an order-one quantum correction to the classical approximation
cf. \cite{Frank-Seiringer},\cite{Gross}. These higher-order asymptotics
should be provable with better control of the error terms when arguing
the lower bound, cf. \cite{Frank-Seiringer}, and by making full use
of the one-dimensional effective Coulomb potential: the electron-phonon
interaction of the one-dimensional Hamiltonian that is derived both
here and in Frank et al.'s proof \cite{Frank-Geisinger} of the lower
bound is described using an artificial coupling function, given in
(\ref{eq:Coupling Function}) below, when instead it should really
be argued that the coupling is proportional to the square-root of
the Fourier transform of the effective Coulomb potential in (\ref{eq:uppercut});
then the strategy from \cite{Lieb Thomas} can be used to arrive at
the classical approximation in (\ref{eq:classic}) now as a lower
bound. 

The two-term asymptotics of the ground-state energy achieved in Theorem
\ref{thm:Energy Asymptotics} suffices for arguing the main result: 
\begin{thm}
\label{thm:Wavefunction} Let $\Psi^{(B)}\in\mathcal{H}$ be any approximate
ground-state wave function satisfying $\left\langle \Psi^{(B)},\mathbb{H}(B)\Psi^{(B)}\right\rangle =E_{0}(B)+\mathfrak{o}(\left(\ln B\right)^{2})$.
The one-dimensional minimization problem in (\ref{eq: p}) admits
up to complex phase a unique minimizer 
\[
\phi_{0}\left(x_{3}\right)=\frac{\alpha+2\beta}{\sqrt{8\alpha}\cosh\left(\left(\frac{\alpha+2\beta}{4}\right)\left|x_{3}\right|+\tanh^{-1}\left(\frac{2\beta}{\alpha+2\beta}\right)\right)},
\]
and for $W$ a sum of a bounded Borel measure on the real line and
a $L^{\infty}\left(\mathbb{R}\right)$ function 
\begin{align}
 & \lim_{B\rightarrow\infty}\ \frac{1}{\left(\ln B\right)}\int_{\mathbb{R}}W(x_{3})\left(\int_{\mathbb{R}^{2}}\left\Vert \Psi^{(B)}\right\Vert _{\mathcal{F}}^{2}\left(x_{\perp},\frac{x_{3}}{\left(\ln B\right)}\right)dx_{\perp}\right)dx_{3}\nonumber \\
 & =\int_{\mathbb{R}}W(x_{3})\,\phi_{0}(x_{3})^{2}\,dx_{3}.\label{eq:limiting density}
\end{align}
\end{thm}

\noindent By choosing $W$ as a delta-function potential pointwise
convergence is obtained. When $\alpha=0$ the limiting density in
(\ref{eq:limiting density}) is $\sqrt{\left(\beta/2\right)}\exp(-\beta|x_{3}|/2)$
cf. \cite{Froese Waxler},\cite{Rau Spruch}. The idea of the proof
is to add to the Hamiltonian $\epsilon$ times the one-dimensional
potential $W$ scaled appropriately in the magnetic-field direction
cf. \cite{Baumgartner},\cite{Frank-Merz-Siedentop-Simon},\cite{Iatchenko-Lieb-Siedentop},\cite{Lieb-Simon},\cite{Lieb-Solovej-Yngvason},\cite{Lieb-Yau}.
For 
\begin{equation}
\mathbb{H}_{\epsilon}(B):=\mathbb{H}(B)-\epsilon\left(\ln B\right)^{2}W\left(\left(\ln B\right)x_{3}\right)\label{eq:pert}
\end{equation}
and $E_{\epsilon}(B)$ the corresponding ground-state energy it is
argued vis-à-vis Theorem \ref{thm:Energy Asymptotics} that 
\begin{equation}
E_{\epsilon}(B)=B+\mathfrak{e}_{_{\epsilon}}\left(\ln B\right)^{2}+\mathcal{O}\left(\left(\ln B\right)^{3/2}\right)\ \text{as}\ B\rightarrow\infty\label{eq:1-3}
\end{equation}
with 
\begin{equation}
\mathfrak{e}_{_{\epsilon}}:=\inf_{\left\Vert \varphi\right\Vert _{2}=1}\left\{ \int_{\mathbb{R}}\left|\varphi'\right|^{2}dx-\frac{\alpha}{2}\int_{\mathbb{R}}\left|\varphi\right|^{4}dx-\beta\left|\varphi(0)\right|^{2}-\epsilon\int_{\mathbb{R}}W(x)\left|\varphi\right|^{2}dx\right\} .\label{eq:pert2}
\end{equation}
By the variational principle $E_{\epsilon}\leq\left\langle \Psi^{(B)},\mathbb{H_{\epsilon}}(B)\Psi^{(B)}\right\rangle $,
and the expectation value on the right-hand side evaluates to
\[
E_{0}(B)+\mathfrak{o}(\left(\ln B\right)^{2})-\epsilon\left(\ln B\right)\int_{\mathbb{R}}W\left(x_{3}\right)\left(\int_{\mathbb{R}^{2}}\left\Vert \Psi^{(B)}\right\Vert _{\mathcal{F}}^{2}\left(x_{\perp},\frac{x_{3}}{\ln B}\right)dx_{\perp}\right)dx_{3}.
\]
Then for $\epsilon>0$
\[
\frac{E_{0}(B)-E_{\epsilon}(B)}{\epsilon\left(\ln B\right)^{2}}\geq\frac{1}{\left(\ln B\right)}\int_{\mathbb{R}}W\left(x_{3}\right)\left(\int_{\mathbb{R}^{2}}\left\Vert \Psi^{(B)}\right\Vert _{\mathcal{F}}^{2}\left(x_{\perp},\frac{x_{3}}{\ln B}\right)dx_{\perp}\right)dx_{3}+\mathfrak{o}\left(1\right),
\]
and taking the limit $B\rightarrow\infty$ by Theorem \ref{thm:Energy Asymptotics}
and (\ref{eq:1-3}) 
\[
\frac{\mathfrak{e}_{_{0}}-\mathfrak{e}_{_{\epsilon}}}{\epsilon}\geq\limsup_{B\rightarrow\infty}\frac{1}{\left(\ln B\right)}\int_{\mathbb{R}}W\left(x_{3}\right)\left(\int_{\mathbb{R}^{2}}\left\Vert \Psi^{(B)}\right\Vert _{\mathcal{F}}^{2}\left(x_{\perp},\frac{x_{3}}{\ln B}\right)dx_{\perp}\right)dx_{3}.
\]

\noindent For $\epsilon<0$ the above inequality is reversed with
``lim sup'' replaced by ``lim inf.'' Hence the theorem follows
if the quotient on the left-hand side has a limit as $\epsilon\rightarrow0$.
The map $\epsilon\mapsto\mathfrak{e}_{_{\epsilon}}$ is differentiable
at $\epsilon=0$ for all $W$ if and only if the one-dimensional problem
for the energy $\mathfrak{e}_{_{0}}$ in (\ref{eq: p}) admits up
to complex phase a unique minimizer: Uniqueness is established by
explicitly solving the corresponding Euler-Lagrange equation cf. \cite{Lieb-Solovej-Yngvason},
and by the variational principle and a compactness argument 
\begin{equation}
\lim_{\epsilon\rightarrow0}\frac{\mathfrak{e}_{_{0}}-\mathfrak{e}_{_{\epsilon}}}{\epsilon}=\int_{\mathbb{R}}W(x_{3})\,\phi_{0}(x_{3})^{2}\,dx_{3}.\label{eq:diffurentiation}
\end{equation}

In Section 3 the differentiation of the one-dimensional energy (\ref{eq:diffurentiation})
is proved as Theorem \ref{thm:Differentiation}. In Section 4 an upper
bound to the ground-state energy is proved as Theorem \ref{Upper Bound}.
In Section 5 a lower bound to the ground-state energy is proved as
Theorem \ref{thm:Lower Bound}, and Theorem \ref{thm:Energy Asymptotics}
follows from Theorem \ref{Upper Bound} and Theorem \ref{thm:Lower Bound}.
In Section \ref{sec:Proof-of-Theorem 2.2} the main result Theorem
\ref{thm:Wavefunction} is proved.

\section{Differentiating the One-Dimensional Energy}

\noindent The one-dimensional problem for the energy $\mathfrak{e}_{_{0}}$
in (\ref{eq: p}) shall be denoted 
\[
\mathfrak{e}_{_{0}}:=\inf_{\left\Vert \varphi\right\Vert _{2}=1}\mathcal{E}_{_{0}}\left(\varphi\right)
\]
with 
\begin{equation}
\mathcal{E}_{_{0}}\left(\varphi\right):=\int_{\mathbb{R}}\left|\varphi'\right|^{2}dx-\frac{\alpha}{2}\int_{\mathbb{R}}\left|\varphi\right|^{4}dx-\beta\left|\varphi(0)\right|^{2}.\label{eq: Pekar Functional}
\end{equation}

\begin{lem}
The minimization problem in (\ref{eq: p}) for the energy $\mathfrak{e}_{_{0}}$
admits up to complex phase a unique minimizer 
\begin{equation}
\phi_{_{0}}\left(x\right)=\frac{\alpha+2\beta}{\sqrt{8\alpha}\cosh\left(\left(\frac{\alpha+2\beta}{4}\right)\left|x\right|+\tanh^{-1}\left(\frac{2\beta}{\alpha+2\beta}\right)\right)}\label{eq:min-min}
\end{equation}
and 
\[
\mathfrak{e}_{_{0}}=-\frac{1}{48}\left(\alpha^{2}+6\alpha\beta+12\beta^{2}\right).
\]
\end{lem}

\begin{proof}
The existence of a minimizer is shown in Appendix A. Any minimizer
up to multiplication by a complex phase is a nonnegative, $C^{2}\left(\mathbb{R}\backslash\left\{ 0\right\} \right)$
function in $H^{1}\left(\mathbb{R}\right)$ solving the Euler-Lagrange
equation $-\psi''-\alpha\psi^{3}-\beta\delta(x)\psi=-\lambda\psi$
with
\begin{equation}
-\lambda=\mathfrak{e}_{_{0}}-\frac{\alpha}{2}\int_{\mathbb{R}}\psi^{4}dx<0.\label{eq:lambda}
\end{equation}
Or equivalently it must solve 
\begin{equation}
-\psi''-\alpha\psi^{3}=-\lambda\psi\ \ \text{for}\ \left|x\right|>0\label{eq:u1}
\end{equation}
and satisfy the boundary condition 
\begin{equation}
\lim_{\epsilon\rightarrow0^{+}}\left[\psi'\left(-\epsilon\right)-\psi'(\epsilon)\right]=\beta\psi(0).\label{eq:u2}
\end{equation}
The first integral of (\ref{eq:u1}) is 
\begin{equation}
\psi'^{2}=-\frac{\alpha}{2}\psi^{4}+\lambda\psi^{2}\ \ \text{for}\ \left|x\right|>0.\label{eq:u3}
\end{equation}
Any nonnegative, $C^{2}\left(\mathbb{R}\backslash\{0\}\right)$ solution
of (\ref{eq:u3}) in $H^{1}(\mathbb{R})$ satisfying the boundary
condition in (\ref{eq:u2}) must be of the form 
\begin{equation}
\psi=\sqrt{\frac{2\lambda}{\alpha}}\frac{1}{\cosh\left(\sqrt{\lambda}\left(\left|x\right|-\tau\right)\right)}\label{form}
\end{equation}
for some $\tau$ cf. \cite{Frank Luminy}. The boundary condition
and that $\left\Vert \psi\right\Vert _{2}=1$ require respectively
\[
\tanh\left(\tau\sqrt{\lambda}\right)=-\frac{\beta}{2\sqrt{\lambda}}\ \ \ \text{and}\ \ \ \frac{\alpha}{4\sqrt{\lambda}}=1+\tanh\left(\tau\sqrt{\lambda}\right).
\]
Any minimizer up to complex phase must therefore be of the form in
(\ref{form}) with 
\[
\lambda=\left(\frac{\alpha+2\beta}{4}\right)^{2}\ \ \ \text{and}\ \ \ \tau=-\frac{4}{\alpha+2\beta}\tanh^{-1}\left(\frac{2\beta}{\alpha+2\beta}\right).
\]
The explicit calculation of $\mathfrak{e}_{_{0}}$ now follows from
(\ref{eq:lambda}). 
\end{proof}
\begin{lem}
\label{lem:Every-minimizing-sequence}If $\left\{ \phi_{_{n}}\right\} _{n=1}^{\infty}$
is a minimizing sequence for $\mathfrak{e}_{_{0}}$, then $\left|\phi_{_{n}}\right|$
converges in $H^{1}\left(\mathbb{R}\right)$ to the minimizer $\phi_{_{0}}$
given in (\ref{eq:min-min}). 
\end{lem}

\begin{proof}
\noindent By Theorem 7.8 in \cite{Lieb-Loss} $\left\{ \left|\phi_{_{n}}\right|\right\} _{n=1}^{\infty}$
is also a minimizing sequence for $\mathfrak{e}_{_{0}}$. It is argued
in Appendix A that every minimizing sequence for $\mathfrak{e}_{_{0}}$
has a subsequence converging in $H^{1}(\mathbb{R})$ to some minimizer.
Since $\phi_{_{0}}$ is up to complex phase the unique minimizer,
every subsequence of $\left\{ \left|\phi_{_{n}}\right|\right\} _{n=1}^{\infty}$
must converge in $H^{1}\left(\mathbb{R}\right)$ to $\phi_{_{0}}$. 
\end{proof}
\begin{thm}
\label{thm:Differentiation}Let $W$ be a sum of a bounded Borel measure
on the real line and a $L^{\infty}\left(\mathbb{R}\right)$ function.
For $\epsilon$ a real parameter consider the one-dimensional energy
\begin{equation}
\mathfrak{e}_{_{\epsilon}}:=\inf_{\left\Vert \varphi\right\Vert _{2}=1}\mathcal{E}_{_{\epsilon}}\left(\varphi\right),\label{eq:perturbation}
\end{equation}
where 
\begin{equation}
\mathcal{E}_{_{\epsilon}}\left(\varphi\right):=\mathcal{E}_{_{0}}\left(\varphi\right)-\epsilon\int_{\mathbb{R}}W\left(x\right)\left|\varphi\left(x\right)\right|^{2}dx\label{eq:perturbed functional}
\end{equation}
with the functional $\mathcal{E}_{_{0}}$ as given in (\ref{eq: Pekar Functional}).
Then the map $\epsilon\mapsto\mathfrak{e}_{_{\epsilon}}$ is differentiable
at $\epsilon=0$ and 
\[
\left.\frac{d}{d\epsilon}\right|_{\epsilon=0}\mathfrak{e}_{_{\epsilon}}=-\int_{\mathbb{R}}W(x)\,\phi_{0}(x)^{2}\,dx
\]
with $\phi_{_{0}}$ the minimizer given in (\ref{eq:min-min}) for
the energy $\mathfrak{e}_{_{0}}$. 
\end{thm}

\begin{proof}
Writing $W=\mu+\omega$ where $\mu$ is a signed, bounded measure
on the real line and $\omega\in L^{\infty}\left(\mathbb{R}\right)$,
it follows from Hölder's inequality, the Sobolev inequality and completion
of the square that for $\varphi\in H^{1}\left(\mathbb{R}\right)$
\begin{align}
\mathcal{E}_{_{\epsilon}}\left(\varphi\right) & \geq\|\varphi^{\prime}\|_{2}^{2}-\|\varphi\|_{\infty}^{2}\left(\frac{\alpha}{2}\|\varphi\|_{2}^{2}+\left|\epsilon\right|\left|\mu\right|\left(\mathbb{R}\right)+\beta\right)-\left|\epsilon\right|\|\omega\|_{\infty}\|\varphi\|_{2}^{2}\nonumber \\
 & \geq\frac{3}{4}\|\varphi^{\prime}\|_{2}^{2}-\|\varphi\|_{2}^{2}\left(\frac{\alpha}{2}\|\varphi\|_{2}^{2}+\left|\epsilon\right|\left|\mu\right|\left(\mathbb{R}\right)+\beta\right)^{2}-\left|\epsilon\right|\|\omega\|_{\infty}\|\varphi\|_{2}^{2}.\label{eq:p1}
\end{align}
Hence $\mathfrak{e}_{_{\epsilon}}>-\infty$, and for each $\epsilon$
there is a $\phi_{_{\epsilon}}\in H^{1}\left(\mathbb{R}\right)$ satisfying
\begin{equation}
\mathcal{E}_{_{\epsilon}}\left(\phi_{_{\epsilon}}\right)\leq\mathfrak{e}_{_{\epsilon}}+\epsilon^{2}\ \text{and}\ \left\Vert \phi_{_{\epsilon}}\right\Vert _{_{2}}=1.\label{eq:d5}
\end{equation}
By the variational principle 
\[
\mathfrak{e}_{_{0}}\leq\mathcal{E}_{_{0}}\left(\phi_{_{\epsilon}}\right)=\mathcal{E}_{_{\epsilon}}\left(\phi_{_{\epsilon}}\right)+\epsilon\int_{\mathbb{R}}W\left(x\right)\left|\phi_{_{\epsilon}}\left(x\right)\right|^{2}dx\leq\mathfrak{e}_{_{\epsilon}}+\epsilon^{2}+\epsilon\int_{\mathbb{R}}W\left(x\right)\left|\phi_{_{\epsilon}}\left(x\right)\right|^{2}dx,
\]
and 
\[
\mathfrak{e}_{_{\epsilon}}\leq\mathcal{E}_{_{\epsilon}}\left(\phi_{_{0}}\right)=\mathfrak{e}_{_{0}}-\epsilon\int_{\mathbb{R}}W(x)\,\phi_{0}(x)^{2}\,dx.
\]
Then for $\epsilon>0$ 
\begin{equation}
-\int_{\mathbb{R}}W\left(x\right)\left|\phi_{_{\epsilon}}\left(x\right)\right|^{2}dx-\epsilon\leq\frac{\mathfrak{e}_{_{\epsilon}}-\mathfrak{e}_{_{0}}}{\epsilon}\leq-\int_{\mathbb{R}}W(x)\,\phi_{0}(x)^{2}\,dx.\label{eq:d8}
\end{equation}
For $\epsilon<0$ the inequalities in (\ref{eq:d8}) are reversed.
It suffices therefore to show for any sequence $\left\{ \epsilon_{n}\right\} _{n=1}^{\infty}$,
$\left|\epsilon_{n}\right|>0$ and $\epsilon_{n}\rightarrow0$ that
\begin{equation}
\lim_{n\rightarrow\infty}\int_{\mathbb{R}}W\left(x\right)\left|\phi_{_{\epsilon_{n}}}\left(x\right)\right|^{2}dx=\int_{\mathbb{R}}W(x)\,\phi_{0}(x)^{2}\,dx.\label{eq:d9}
\end{equation}

Since $\mathfrak{e}_{_{\epsilon}}>-\infty$, the concave map $\epsilon\mapsto\mathfrak{e}_{_{\epsilon}}$
is continuous. It follows from (\ref{eq:p1}), the Sobolev inequality
and (\ref{eq:d5}) that $\|\phi_{_{\epsilon_{n}}}\|_{_{\infty}}<C$
. Thus
\[
\lim_{n\rightarrow\infty}\epsilon_{n}\int_{\mathbb{R}}W\left(x\right)\left|\phi_{_{\epsilon_{n}}}\left(x\right)\right|^{2}dx=0,
\]
\begin{align*}
\mathfrak{e}_{_{0}}=\lim_{n\rightarrow\infty}\mathfrak{e}_{_{\epsilon_{n}}} & \geq\limsup_{n\rightarrow\infty}\left(\mathcal{E}_{_{0}}\left(\phi_{_{\epsilon_{n}}}\right)-\epsilon_{n}\int_{\mathbb{R}}W\left(x\right)\left|\phi_{_{\epsilon_{n}}}\left(x\right)\right|^{2}dx-\epsilon_{n}^{2}\right)\\
 & =\limsup_{n\rightarrow\infty}\,\mathcal{E}_{_{0}}\left(\phi_{_{\epsilon_{n}}}\right)
\end{align*}
and $\left\{ \phi_{_{\epsilon_{n}}}\right\} _{n=1}^{\infty}$ is a
minimizing sequence for $\mathfrak{e}_{_{0}}$. By Lemma \ref{lem:Every-minimizing-sequence}
$\left|\phi_{_{\epsilon_{n}}}\right|$ converges in $H^{1}\left(\mathbb{R}\right)$
to $\phi_{_{0}}$, and by Theorem 8.6 in \cite{Lieb-Loss} $\left|\phi_{_{\epsilon_{n}}}\right|$
converges also pointwise uniformly to $\phi_{_{0}}$ on bounded sets.
The convergence in (\ref{eq:d9}) now follows. 
\end{proof}

\section{Upper Bound to the Ground-State Energy}
\begin{thm}
\label{Upper Bound}There is a constant $C>0$ such that for $B>1$

\noindent 
\[
E_{0}(B)\leq B+\mathfrak{e}_{_{0}}\left(\ln B\right)^{2}-4\mathfrak{e}_{_{0}}\ln B\ln\ln B+C\ln B.
\]
\end{thm}

\noindent Theorem \ref{Upper Bound} will follow from Lemma \ref{lem:Classical Pekar}
and Lemma \ref{lem:Last Straw}. 
\begin{lem}
\label{lem:Classical Pekar}$E_{0}(B)\leq E_{0}^{c}(B)$, where $E_{0}^{c}(B):=\inf\left\{ \mathcal{P}\left(\psi\right)\colon\|\psi\|_{2}=1\right\} $
is the classical Pekar energy with $\mathcal{P}$ denoting the three-dimensional
magnetic Pekar functional
\[
\left\langle \psi,\left(H_{B}-\partial_{3}^{2}\right)\psi\right\rangle _{L^{^{2}}}-\frac{\alpha}{2}\int_{\mathbb{R}^{3}}\int_{\mathbb{R}^{3}}\frac{\left|\psi\left(x\right)\right|^{2}\left|\psi\left(y\right)\right|^{2}}{\left|x-y\right|}dxdy-\beta\int_{\mathbb{R}^{3}}\frac{\left|\psi\left(x\right)\right|^{2}}{\left|x\right|}dx.
\]
\end{lem}

\begin{proof}
By the variational principle 
\[
E_{0}(B)\leq\inf\left\{ \left\langle \Psi,\,\mathbb{H}(B)\Psi\right\rangle :\ \|\Psi\|=1,\,\Psi=\varphi(x)\Phi\ \text{and}\ \Phi\in\mathcal{F}\right\} =E_{0}^{c}(B),
\]
the equality following from completion of the square. 
\end{proof}
\begin{lem}
\label{lem:Last Straw}There is a constant $C>0$ such that for $B>1$
\[
E_{0}^{c}\left(B\right)\leq B+\mathfrak{e}_{_{0}}\left(\ln B\right)^{2}-4\mathfrak{e}_{_{0}}\ln B\ln\ln B+C\ln B.
\]
\end{lem}

\begin{proof}
With the lowest Landau state in the zero angular momentum sector 
\[
\gamma_{B}\left(x_{\perp}\right):=\sqrt{\frac{B}{2\pi}}\exp\left(-\frac{B}{4}\left|x_{\perp}\right|^{2}\right)\ \text{i.e.}\ H_{B}\,\gamma_{B}=B\,\gamma_{B},
\]
by an elementary calculation \cite{Brummelhuis-Ruskai} 
\begin{equation}
V_{\mathcal{U}}^{B}(x_{3}):=\int_{\mathbb{R}^{2}}\frac{\left|\gamma_{B}\left(x_{\perp}\right)\right|^{2}}{\sqrt{\left|x_{\perp}\right|^{2}+x_{3}^{2}}}dx_{\perp}=\int_{0}^{\infty}\frac{e^{-u}}{\sqrt{x_{3}^{2}+\frac{2u}{B}}}\,du\label{eq:uppercut}
\end{equation}
and 
\begin{equation}
\int_{\mathbb{R}^{2}}\int_{\mathbb{R}^{2}}\frac{\left|\gamma_{B}\left(x_{\perp}\right)\right|^{2}\left|\gamma_{B}\left(y_{\perp}\right)\right|^{2}}{\sqrt{\left|x_{\perp}-y_{\perp}\right|^{2}+\left(x_{3}-y_{3}\right)^{2}}}dx_{\perp}dy_{\perp}=\frac{1}{\sqrt{2}}V_{\mathcal{U}}^{B}\left(\frac{x_{3}-y_{3}}{\sqrt{2}}\right).\label{eq:effective self-interaction}
\end{equation}
For $L>0$ and with $\phi_{0}$ the minimizer for the energy $\mathfrak{e}_{_{0}}$,
$\mu(B):=\ln B-2\ln\ln B$ and $f_{B}(x_{3}):=\sqrt{\mu(B)}\phi_{0}(\mu(B)x_{3})$
by the variational principle, Lemma \ref{lem:B-label1} and Corollary
\ref{cor:B-2} 
\begin{align}
E_{0}^{c}(B) & \leq\mathcal{P}\left(\gamma_{B}\left(x_{\perp}\right)f_{B}\left(x_{3}\right)\right)\nonumber \\
 & =B+\int_{\mathbb{R}}\left|f_{B}^{\prime}\right|^{2}dx_{3}-\frac{\alpha}{\sqrt{8}}\int\int_{\mathbb{R}\times\mathbb{R}}\left|f_{B}(x_{3})\right|^{2}V_{\mathcal{U}}^{B}\left(\frac{x_{3}-y_{3}}{\sqrt{2}}\right)\left|f_{B}(y_{3})\right|^{2}dx_{3}\,dy_{3}\nonumber \\
 & \,\ \ \ \ \ \ -\beta\int_{\mathbb{R}}V_{\mathcal{U}}^{B}\left(x_{3}\right)\left|f_{B}(x_{3})\right|^{2}dx_{3}\label{eq:classic}\\
 & \leq B+\int_{\mathbb{R}}\left|f_{B}^{\prime}\right|^{2}dx_{3}-\left(\alpha/2\right)\mu(B)\int_{\mathbb{R}}\left|f_{B}(x_{3})\right|^{4}dx_{3}-\beta\mu(B)\left|f_{B}(0)\right|^{2}\nonumber \\
 & \,\ \ \ \ \ \ +\left(\alpha/2+\beta\right)\left(L^{-1}+8\sqrt{L}\left\Vert f_{B}^{\prime}\right\Vert _{2}^{3/2}+\left|\mathcal{G}\left(B,L/\sqrt{2}\right)\right|\left\Vert f_{B}^{\prime}\right\Vert _{2}\right)\nonumber 
\end{align}
where $\mathcal{G}\left(B,L\right):=2\ln L+2\ln\ln B+2\int_{0}^{\infty}e^{-u}\ln\left(\sqrt{\frac{1}{u}+\frac{2}{BL^{2}}}+\sqrt{\frac{1}{u}}\right)du-\ln2.$
Now choosing $L=1/\ln B$ it can be verified $\left|\mathcal{G}\left(B,L/\sqrt{2}\right)\right|<C$
for $B>1$. Since $\left\Vert f_{B}^{\prime}\right\Vert _{2}=\mu(B)\left\Vert \phi_{0}^{\prime}\right\Vert _{2}$
and 
\begin{equation}
\int_{\mathbb{R}}\left|f_{B}^{\prime}\right|^{2}dx_{3}-\left(\alpha/2\right)\mu(B)\int\left|f_{B}\right|^{4}dx_{3}-\beta\mu(B)\left|f_{B}(0)\right|^{2}=\left(\mu(B)\right)^{2}\mathfrak{e}_{_{0}},\label{eq:scaling}
\end{equation}
the lemma follows. 
\end{proof}
\begin{cor}
\label{cor:Upper}Let $W$ be a sum of a bounded Borel measure on
the real line and a $L^{\infty}(\mathbb{R})$ function. For $\epsilon$
a real parameter, $E_{\epsilon}(B)$ the ground-state energy of the
Hamiltonian $\mathbb{H}_{\epsilon}(B)$ in (\ref{eq:pert}) and $\mathfrak{e}_{_{\epsilon}}$
the one-dimensional energy in (\ref{eq:pert2}) there is a constant
$C>0$ such that for $B>1$ 
\[
E_{\epsilon}(B)\leq B+\mathfrak{e}_{_{\epsilon}}\left(\ln B\right)^{2}+C\ln B\left|\ln\ln B\right|+C\ln B.
\]
\end{cor}

\begin{proof}
By the estimate in (\ref{eq:p1}) $\mathfrak{e}_{_{\epsilon}}>-\infty$.
Then with the functional $\mathcal{E}_{_{\epsilon}}$ as given in
(\ref{eq:perturbed functional}) for $B>1$ there is a $\phi_{B}\in H^{1}\left(\mathbb{R}\right)$,
$\left\Vert \phi_{B}\right\Vert _{2}=1$ satisfying 
\begin{equation}
\mathcal{E}_{_{\epsilon}}\left(\phi_{B}\right)<\mathfrak{e}_{_{\epsilon}}+1/\ln B\ \text{and}\ \text{\ensuremath{\|}\ensuremath{\ensuremath{\phi_{B}^{\prime}\|_{2}}<C.}}\label{eq:approximate minimizers}
\end{equation}
For $L>0$ and with $g_{B}\left(x_{3}\right):=\sqrt{\ln B}\phi_{B}\left(\left(\ln B\right)x_{3}\right)$
by trivial modifications to Lemma \ref{lem:B-label1} and Corollary
\ref{cor:B-2} the arguments in the proofs of Lemma \ref{lem:Classical Pekar}
and Lemma \ref{lem:Last Straw} apply mutatis mutandis, and $E_{\epsilon}\left(B\right)\leq E_{\epsilon}^{c}(B)$
with
\begin{align*}
 & E_{\epsilon}^{c}(B):=\underset{\|\psi\|_{_{2}}=1}{\inf}\left\{ \mathcal{P}\left(\psi\right)-\epsilon\left(\ln B\right)^{2}\int_{\mathbb{R}^{3}}W\left(\left(\ln B\right)x_{3}\right)\left|\psi\left(x_{\perp},x_{3}\right)\right|^{2}dx_{\perp}dx_{3}\right\} \\
 & \leq B+\int_{\mathbb{R}}\left|g_{B}^{\prime}\right|^{2}dx_{3}-\left(\alpha/2\right)\ln B\int_{\mathbb{R}}\left|g_{B}\right|^{4}dx_{3}-\beta\ln B\left|g_{B}(0)\right|^{2}\\
 & \ \ \ \ \ \ -\epsilon\left(\ln B\right)^{2}\int_{\mathbb{R}}W\left(\left(\ln B\right)x_{3}\right)\left|g_{B}\left(x_{3}\right)\right|^{2}dx_{3}\\
 & \ \ \ \ \ \ +\left(\alpha/2+\beta\right)\left(L^{-1}+8\sqrt{L}\left\Vert g_{B}^{\prime}\right\Vert _{2}^{3/2}+\left|\tilde{\mathcal{G}}\left(B,L/\sqrt{2}\right)\right|\left\Vert g_{B}^{\prime}\right\Vert _{2}\right)
\end{align*}
where $\tilde{\mathcal{G}}\left(B,L\right):=2\ln L+2\int_{0}^{\infty}e^{-u}\ln\left(\sqrt{\frac{1}{u}+\frac{2}{BL^{2}}}+\sqrt{\frac{1}{u}}\right)du-\ln2$.
Now choosing $L=1/\ln B$ it can be verified $|\tilde{\mathcal{G}}\left(B,L/\sqrt{2}\right)|<2\left|\ln\ln B\right|+C$
for $B>1$. Since $\|g_{B}^{\prime}\|_{2}=\left(\ln B\right)\left\Vert \phi_{B}^{\prime}\right\Vert _{2}$,
the corollary follows from (\ref{eq:approximate minimizers}) and
scaling as in (\ref{eq:scaling}). 
\end{proof}

\section{Lower Bound to the Ground-State Energy}
\begin{thm}
\label{thm:Lower Bound}There is a constant $C>0$ such that for $B\geq C$
\[
E_{0}(B)\geq B+\mathfrak{e}_{_{0}}\left(\ln B\right)^{2}-C\left(\ln B\right)^{3/2}.
\]
\end{thm}

\noindent The proof of Theorem \ref{thm:Lower Bound} shall be provided
at the end of Subsection \ref{subsec:Lowest-Landau-Level}.

\subsection{Delta-function potential}

For $\Psi\in H_{A}^{1}\left(\mathbb{R}^{3}\right)\otimes\mathcal{F}$
its electron density in the magnetic-field direction shall be denoted
$\widetilde{\Psi}^{2}\left(x_{3}\right)$ i.e. 
\[
\widetilde{\Psi}\left(x_{3}\right):=\left(\int_{\mathbb{R}^{2}}\left\Vert \Psi\right\Vert _{\mathcal{F}}^{2}\left(x_{\perp},x_{3}\right)dx_{\perp}\right)^{1/2}.
\]
By Hölder's inequality $\|\partial_{3}\widetilde{\Psi}\|_{2}\leq\|\partial_{3}\Psi\|$
and $\tilde{\Psi}\in H^{1}\left(\mathbb{R}\right)$. Furthermore the
integral operator $P_{0}^{B}$ acting on $L^{2}\left(\mathbb{R}^{2}\right)$
with kernel 
\begin{equation}
P_{0}^{B}(x_{\perp},y_{\perp}):=\frac{B}{2\pi}e^{-\frac{B}{4}\left|x_{\perp}-y_{\perp}\right|^{2}}e^{\frac{iB}{2}\left(x_{1}y_{2}-x_{2}y_{1}\right)}\label{eq:kernel}
\end{equation}
is the projection onto the lowest Landau level i.e. the ground state
of the Landau Hamiltonian $H_{B}$, and $P_{>}^{B}:=1-P_{0}^{B}$.
Below the operators $P_{0}^{B}\otimes1$ and $P_{0}^{B}\otimes1\otimes1$
acting respectively on $L^{2}\left(\mathbb{R}^{2}\right)\otimes L^{2}\left(\mathbb{R}\right)$
and $L^{2}\left(\mathbb{R}^{2}\right)\otimes L^{2}\left(\mathbb{R}\right)\otimes\mathcal{F}$
shall also be denoted $P_{0}^{B}$. 
\begin{lem}
\label{Lemma: Bathtub}Let $L>0$. For $B>1$ and $\Psi\in H_{A}^{1}(\mathbb{R}^{3})\otimes\mathcal{F}$
\begin{align*}
 & \int_{\mathbb{R}^{3}}\frac{\left\Vert P_{0}^{B}\Psi\right\Vert _{\mathcal{F}}^{2}\left(x_{\perp},x_{3}\right)}{\sqrt{\left|x_{\perp}\right|^{2}+x_{3}^{2}}}\,dx_{\perp}dx_{3}-\left(\ln B-2\ln\ln B\right)\left(\widetilde{\Psi}(0)\right)^{2}\\
 & \leq L^{-1}\|\widetilde{\Psi}\|_{2}^{2}+8\sqrt{L}\|\partial_{3}\widetilde{\Psi}\|_{2}^{3/2}\|\widetilde{\Psi}\|_{2}^{1/2}+\left|\mathcal{D}\left(B,L\right)\right|\|\partial_{3}\widetilde{\Psi}\|_{2}\|\widetilde{\Psi}\|_{2};
\end{align*}
\begin{equation}
\mathcal{D}\left(B,L\right):=2\ln L+2\ln\ln B+\frac{2}{\sqrt{\frac{2}{BL^{2}}+1}+1}+2\ln\left(\sqrt{1+\frac{2}{BL^{2}}}+1\right)-\ln2.\label{eq:D(B,L)}
\end{equation}
\end{lem}

\begin{proof}
\noindent By Hölder's inequality 
\begin{equation}
\left\Vert P_{0}^{B}\Psi\right\Vert _{\mathcal{F}}^{2}\left(x_{\perp},x_{3}\right)\leq\left(\frac{B}{2\pi}\right)\left(\widetilde{\Psi}\left(x_{3}\right)\right)^{2},\label{eq:1}
\end{equation}
and since $P_{0}^{B}$ is a projection 
\begin{equation}
\int_{\mathbb{R}^{3}}\left\Vert P_{0}^{B}\Psi\right\Vert _{\mathcal{F}}^{2}\left(x_{\perp},x_{3}\right)dx_{\perp}dx_{3}\leq\int_{\mathbb{R}}\left(\widetilde{\Psi}\left(x_{3}\right)\right)^{2}dx_{3}.\label{eq:2}
\end{equation}
It follows from the bathtub principle \cite{Lieb-Loss},\cite{Bonetto-Loss}
that the maximum of the expression 
\[
\int_{\mathbb{R}^{2}}\frac{G\left(x_{\perp},x_{3}\right)}{\sqrt{\left|x_{\perp}\right|^{2}+x_{3}}}\,dx_{\perp}
\]
over all functions $G$ satisfying the conditions (\ref{eq:1}) and
(\ref{eq:2}) above is attained by the function 
\[
G_{\text{max}}\left(x_{\perp},x_{3}\right)=\left\{ \begin{array}{ccc}
\left(\frac{B}{2\pi}\right)\left(\widetilde{\Psi}\left(x_{3}\right)\right)^{2} & \text{when} & \left|x_{\perp}\right|\leq R\\
\\
0 & \text{when} & \left|x_{\perp}\right|>R
\end{array}\right.
\]
where $R=\sqrt{2/B}$. Therefore 
\begin{align*}
\int_{\mathbb{R}}\int_{\mathbb{R}^{2}}\frac{\left\Vert P_{0}^{B}\Psi\right\Vert _{\mathcal{F}}^{2}\left(x_{\perp},x_{3}\right)}{\sqrt{\left|x_{\perp}\right|^{2}+x_{3}^{2}}}dx_{\perp}dx_{3} & \leq\int_{\mathbb{R}}\int_{\left|x_{\perp}\right|\leq\sqrt{2/B}}\ \frac{G_{\text{max}}\left(x_{\perp},x_{3}\right)}{\sqrt{\left|x_{\perp}\right|^{2}+x_{3}^{2}}}dx_{\perp}dx_{3}\\
 & =\int_{\mathbb{R}}V_{\mathcal{\mathcal{L}}}^{B}(x_{3})\left(\widetilde{\Psi}\left(x_{3}\right)\right)^{2}dx_{3}
\end{align*}
with 
\begin{equation}
V_{\mathcal{\mathcal{L}}}^{B}(x_{3}):=\frac{2}{\sqrt{\frac{2}{B}+x_{3}^{2}}+\left|x_{3}\right|}.\label{eq:Bath}
\end{equation}
The lemma follows from Corollary B.3 in the appendix. 
\end{proof}
\begin{cor}
\label{Hydrogen}Let $\tau>0$. There is a constant $C>0$ such that
for $B\geq C$ and $\Psi\in H_{A}^{1}\left(\mathbb{R}^{3}\right)\otimes\mathcal{F}$
\[
\|\partial_{3}\Psi\|^{2}-\tau\left\langle \Psi,\,P_{0}^{B}\left|x\right|^{-1}P_{0}^{B}\Psi\right\rangle \geq\left(-\frac{\tau^{2}}{4}\left(\ln B\right)^{2}+\tau^{2}\ln B\ln\ln B-C\left(\ln B\right)\right)\left\Vert \Psi\right\Vert ^{2}
\]
\end{cor}

\begin{proof}
Recall for $\epsilon>0$, $\|\partial_{3}\widetilde{\Psi}\|_{2}\|\widetilde{\Psi}\|_{2}\leq\epsilon\|\partial_{3}\widetilde{\Psi}\|_{2}^{2}+\epsilon^{-1}\|\widetilde{\Psi}\|_{2}^{2}$.
With $\mathcal{D}\left(B,L\right)$ as given in (\ref{eq:D(B,L)}),
under the assumption that $2\tau\epsilon\left|\mathcal{D}\left(B,L\right)\right|<1/2$
and denoting $\mu(B)=\left(\ln B-2\ln\ln B\right)$ it follows from
Lemma \ref{Lemma: Bathtub} that
\begin{align*}
 & \|\partial_{3}\Psi\|^{2}-\tau\left\langle \Psi,\,P_{0}^{B}\left|x\right|^{-1}P_{0}^{B}\Psi\right\rangle \\
 & \geq\left(1-2\tau\epsilon\left|\mathcal{D}\left(B,L\right)\right|\right)\|\partial_{3}\widetilde{\Psi}\|_{2}^{2}-\tau\mu(B)\left(\widetilde{\Psi}(0)\right)^{2}-\tau\left(L^{-1}+\epsilon^{-1}\left|\mathcal{D}\left(B,L\right)\right|\right)\|\widetilde{\Psi}\|_{2}^{2}\\
 & \ \ \ \,+\tau\epsilon\left|\mathcal{D}\left(B,L\right)\right|\|\partial_{3}\widetilde{\Psi}\|_{2}^{2}-8\tau\sqrt{L}\|\widetilde{\Psi}\|_{2}^{1/2}\|\partial_{3}\widetilde{\Psi}\|_{2}^{3/2}\\
 & \geq\frac{-\tau^{2}\left(\mu(B)\right)^{2}}{4\left(1-2\tau\epsilon\left|\mathcal{D}\left(B,L\right)\right|\right)}\|\widetilde{\Psi}\|_{2}^{2}-\tau\left(\frac{432L^{2}}{\left|\mathcal{D}\left(B,L\right)\right|^{3}\epsilon^{3}}+L^{-1}+\epsilon^{-1}\left|\mathcal{D}\left(B,L\right)\right|\right)\|\widetilde{\Psi}\|_{2}^{2}\\
 & \geq\left(\left(-\frac{\tau^{2}}{4}-\tau^{3}\epsilon\left|\mathcal{D}\left(B,L\right)\right|\right)\left(\mu(B)\right)^{2}-\frac{432\tau L^{2}}{\left|\mathcal{D}\left(B,L\right)\right|^{3}\epsilon^{3}}-\frac{\tau}{L}-\frac{\tau\left|\mathcal{D}\left(B,L\right)\right|}{\epsilon}\right)\|\widetilde{\Psi}\|_{2}^{2}.
\end{align*}
At the second inequality it is used that $\|\partial_{3}\widetilde{\Psi}\|_{2}^{2}-\zeta(\widetilde{\Psi}(0))^{2}\geq-(1/4)\zeta^{2}\|\widetilde{\Psi}\|_{2}^{2}$
for $\zeta\geq0$ and that $ax^{2}-bx^{3/2}\geq-(27/256)b^{4}/a^{3}$
for $a>0,b\geq0$ and $x\geq0$. Now choosing $\epsilon=1/\ln B$
and $L=1/\ln B$ the above assumption that $2\tau\epsilon\left|\mathcal{D}\left(B,L\right)\right|<1/2$
can be verified to be true for $B\geq C$. The lemma follows. 
\end{proof}

\subsection{\label{subsec:Lowest-Landau-Level}Concentration in the Lowest Landau
Level}

First an ultraviolet cutoff is needed. With $\mathcal{K}>0$ and $\Gamma_{\mathcal{K}}:=\left\{ k\in\mathbb{R}^{3}:\max\left(\left|k_{\perp}\right|,\left|k_{3}\right|\right)\leq\mathcal{K}\right\} $
the cutoff Hamiltonian shall be denoted 
\begin{align*}
\mathfrak{h}_{\mathcal{K}}^{\text{co}}:=\left(1-\frac{8\alpha}{\pi\mathcal{K}}\right)\left(H_{B}-\partial_{3}^{2}\right) & +\frac{1}{2}\int_{\Gamma_{\mathcal{K}}}a_{k}^{\dagger}a_{k}\,dk+\frac{1}{2}\mathcal{N}\\
 & +\frac{\sqrt{\alpha}}{2\pi}\int_{\Gamma_{\mathcal{K}}}\left(\frac{a_{k}e^{ik\cdot x}}{|k|}+\frac{a_{k}^{\dagger}e^{-ik\cdot x}}{|k|}\right)dk.
\end{align*}

\begin{lem}
\label{Ultraviolet}If $\mathcal{K}>8\alpha/\pi$, then $\mathbb{H}(B)\geq\mathfrak{h}_{\mathcal{K}}^{\text{co}}-\beta\left|x\right|^{-1}-1/4.$ 
\end{lem}

\begin{proof}
The lemma follows from Lemma 5.1 in \cite{Frank-Geisinger}. 
\end{proof}
\begin{lem}
There exists a constant $C>0$ such that for $B\geq C$ 
\begin{align*}
\mathbb{H}(B)\geq\left(1-\frac{8\alpha\left(\ln B\right)^{2}}{\pi B}\right)H_{B} & +\left(\frac{1}{2}-\frac{8\alpha\left(\ln B\right)^{2}}{\pi B}\right)\left(-\partial_{3}^{2}\right)\\
 & -\beta\left|x\right|^{-1}+\frac{1}{2}\mathcal{N}-C\left(\ln B\right)^{2}.
\end{align*}
\end{lem}

\begin{proof}
Let $\mathcal{K}=B/\left(\ln B\right)^{2}$ and $\mathcal{K}_{3}=16\alpha\left|\ln B\right|/\pi$.
By completion of square
\begin{align}
 & \int_{|k_{3}|\leq\mathcal{K}_{3}}\int_{|k_{\perp}|\leq\mathcal{K}}\left(\frac{1}{2}a_{k}^{\dagger}a_{k}+\frac{\sqrt{\alpha}}{2\pi}\frac{a_{k}}{|k|}e^{ik\cdot x}+\frac{\sqrt{\alpha}}{2\pi}\frac{a_{k}^{\dagger}}{|k|}e^{-ik\cdot x}\right)dk_{\perp}dk_{3}\nonumber \\
 & \geq-\frac{\alpha}{2\pi^{2}}\int_{|k_{3}|\leq\mathcal{K}_{3}}\int_{|k_{\perp}|\leq\mathcal{K}}\frac{1}{|k|^{2}}dk_{\perp}dk_{3}\nonumber \\
 & =-\frac{\alpha}{\pi}\int_{0}^{\mathcal{K}_{3}}\ln\left(\frac{\mathcal{K}^{2}+k_{3}^{2}}{k_{3}^{2}}\right)dk_{3}\nonumber \\
 & =-\frac{\alpha}{\pi}\left(\mathcal{K}_{3}\ln\left(\frac{\mathcal{K}^{2}}{\mathcal{K}_{3}^{2}}+1\right)+2\mathcal{K}\arctan\left(\frac{\mathcal{K}_{3}}{\mathcal{K}}\right)\right)\nonumber \\
 & \geq-\frac{\alpha}{\pi}\mathcal{K}_{3}\left(\ln\left(\frac{\mathcal{K}^{2}}{\mathcal{K}_{3}^{2}}+1\right)+2+\frac{2}{3}\frac{\mathcal{K}_{3}^{2}}{\mathcal{K}^{2}}\right)\nonumber \\
 & \ge-C\left(\ln B\right)^{2},\label{eq:Estimate1}
\end{align}
valid for some constant $C$ and $B\geq C$. The argument corrects
a mistake in the proof of Lemma 5.2 in \cite{Frank-Geisinger}. 

Denoting $\Lambda\left(B\right)=\left\{ \left(k_{\perp},k_{3}\right)\in\Gamma_{\mathcal{K}}:\mathcal{K}_{3}\leq\left|k_{3}\right|\leq\mathcal{K}\ \text{and}\ \left|k_{\perp}\right|\leq\mathcal{K}\right\} $
it can be argued as in the proof of Lemma 5.2 in \cite{Frank-Geisinger}
that for $B$ large
\begin{equation}
\frac{\sqrt{\alpha}}{2\pi}\int_{\Lambda\left(B\right)}\left(\frac{a_{k}}{\left|k\right|}e^{ik\cdot x}+\frac{a_{k}^{\dagger}}{\left|k\right|}e^{-ik\cdot x}\right)dk\geq\frac{1}{2}\partial_{3}^{2}-\frac{1}{2}\left(\int_{\Lambda\left(B\right)}a_{k}^{\dagger}a_{k}dk+\frac{1}{2}\right).\label{eq:Estimate 2}
\end{equation}
The lemma follows from Lemma \ref{Ultraviolet} and the estimates
(\ref{eq:Estimate1}) and (\ref{eq:Estimate 2}). 
\end{proof}
\begin{lem}
\label{Lem: no commute}There exists a constant $C>0$ such that for
$B\geq C$ and all $\Psi\in H_{A}^{1}\left(\mathbb{R}^{3}\right)\otimes\text{dom}\left(\sqrt{\mathcal{N}}\right)$
with $\|\Psi\|=1$ 
\[
\left\langle \Psi,\mathbb{H}(B)\Psi\right\rangle \geq B+\frac{B}{2}\|P_{>}^{B}\Psi\|^{2}+\frac{1}{2}\left\langle \Psi,\,\mathcal{N}\Psi\right\rangle -C\left(\ln B\right)^{2}.
\]
\end{lem}

\begin{proof}
Recall by the diamagnetic inequality for $\tau>0$ 
\[
\left\langle \Psi,\left(H_{B}-\partial_{3}^{2}\right)\Psi\right\rangle \geq\tau\left\langle \Psi,\,\beta|x|^{-1}\Psi\right\rangle -4^{-1}\beta^{2}\tau^{2}\|\Psi\|^{2}.
\]
Then with $0<\eta<1$, $\mathcal{A}>1$ and $\theta:=\left(\frac{1}{2}-\frac{8\alpha\left(\ln B\right)^{2}}{\pi B}\right)$
for $B$ large 
\begin{align*}
 & \theta\left\langle \Psi,\left(H_{B}-\partial_{3}^{2}\right)\Psi\right\rangle -\left\langle \Psi,\,\beta|x|^{-1}\Psi\right\rangle \\
 & =\theta\left\langle P_{0}^{B}\Psi,\left(H_{B}-\partial_{3}^{2}\right)P_{0}^{B}\Psi\right\rangle -\left\langle P_{0}^{B}\Psi,\,\beta|x|^{-1}P_{0}^{B}\Psi\right\rangle \\
 & \ \ +\theta\left(1-\eta\right)\left\langle P_{>}^{B}\Psi,\left(H_{B}-\partial_{3}^{2}\right)P_{>}^{B}\Psi\right\rangle +\theta\eta\left\langle P_{>}^{B}\Psi,\left(H_{B}-\partial_{3}^{2}\right)P_{>}^{B}\Psi\right\rangle \\
 & \ \ -\left\langle P_{>}^{B}\Psi,\,\beta|x|^{-1}P_{>}^{B}\Psi\right\rangle -\left\langle P_{0}^{B}\Psi,\,\beta|x|^{-1}P_{>}^{B}\Psi\right\rangle -\left\langle P_{>}^{B}\Psi,\,\beta|x|^{-1}P_{0}^{B}\Psi\right\rangle \\
 & \geq\theta\left\langle P_{0}^{B}\Psi,\left(H_{B}-\partial_{3}^{2}\right)P_{0}^{B}\Psi\right\rangle -\left\langle P_{0}^{B}\Psi,\,\beta|x|^{-1}P_{0}^{B}\Psi\right\rangle \\
 & \ \ +\theta\left(1-\eta\right)\left\langle P_{>}^{B}\Psi,\left(H_{B}-\partial_{3}^{2}\right)P_{>}^{B}\Psi\right\rangle -\mathcal{A}^{2}\beta^{2}\left(4\theta\eta\right)^{-1}\left\Vert P_{>}^{B}\Psi\right\Vert ^{2}\\
 & \ \ +\left(\mathcal{A}-1\right)\left\langle P_{>}^{B}\Psi,\,\beta|x|^{-1}P_{>}^{B}\Psi\right\rangle -2\left\langle P_{0}^{B}\Psi,\,\beta|x|^{-1}P_{0}^{B}\Psi\right\rangle ^{\frac{1}{2}}\left\langle P_{>}^{B}\Psi,\,\beta|x|^{-1}P_{>}^{B}\Psi\right\rangle ^{\frac{1}{2}}\\
 & \geq\theta\left\langle P_{0}^{B}\Psi,H_{B}P_{0}^{B}\Psi\right\rangle +\left[\theta\left\langle P_{0}^{B}\Psi,-\partial_{3}^{2}P_{0}^{B}\Psi\right\rangle -\left(\mathcal{A}/(\mathcal{A}-1)\right)\left\langle P_{0}^{B}\Psi,\,\beta|x|^{-1}P_{0}^{B}\Psi\right\rangle \right]\\
 & \ \ +\theta\left(1-\eta\right)\left\langle P_{>}^{B}\Psi,\left(H_{B}-\partial_{3}^{2}\right)P_{>}^{B}\Psi\right\rangle -\mathcal{A}^{2}\beta^{2}\left(4\theta\eta\right)^{-1}\left\Vert P_{>}^{B}\Psi\right\Vert ^{2}\\
 & \geq\theta B+\left[\theta\left\langle P_{0}^{B}\Psi,-\partial_{3}^{2}P_{0}^{B}\Psi\right\rangle -\left(\mathcal{A}/(\mathcal{A}-1)\right)\left\langle P_{0}^{B}\Psi,\,\beta|x|^{-1}P_{0}^{B}\Psi\right\rangle \right]\\
 & \ \ +\theta B\left\Vert P_{>}^{B}\Psi\right\Vert ^{2}+\left[\theta\left(B-3B\eta\right)-\mathcal{A}^{2}\beta^{2}\left(4\theta\eta\right)^{-1}\right]\left\Vert P_{>}^{B}\Psi\right\Vert ^{2}.
\end{align*}
Choosing $\eta=1/4$ and $\mathcal{A}=\sqrt{B}\theta/2\beta$ for
$B$ large $\left(\mathcal{A}/(\mathcal{A}-1)\right)<2$, and by Corollary
\ref{Hydrogen} there exists some $C>0$ such that for $B\geq C$
\[
\theta\left\langle P_{0}^{B}\Psi,-\partial_{3}^{2}P_{0}^{B}\Psi\right\rangle -\left(\mathcal{A}/(\mathcal{A}-1)\right)\left\langle P_{0}^{B}\Psi,\,\beta|x|^{-1}P_{0}^{B}\Psi\right\rangle \geq-C\left(\ln B\right)^{2}.
\]
The lemma now follows from Lemma 5.5. 
\end{proof}
\noindent The following observation is immediate from Theorem \ref{Upper Bound}:
For every $M>\mathfrak{e}_{_{0}}$ and $B$ large there exist wave
functions $\Psi\in H_{A}^{1}(\mathbb{R}^{3})\otimes\text{dom}\left(\sqrt{\mathcal{N}}\right)$
satisfying
\begin{equation}
\left\langle \Psi,\mathbb{H}(B)\Psi\right\rangle \leq B+M\left(\ln B\right)^{2}\ \ \text{and}\ \ \|\Psi\|=1.\label{eq:U}
\end{equation}

\begin{cor}
\label{Concentration}For every $M\in\mathbb{R}$ there exists a constant
$C_{M}>0$ such that for $B\ge C_{M}$ and all $\Psi\in H_{A}^{1}(\mathbb{R}^{3})\otimes\text{dom}\left(\sqrt{\mathcal{N}}\right)$
satisfying $\left(\ref{eq:U}\right)$
\[
\|P_{>}^{B}\Psi\|^{2}\leq C_{M}\left(\ln B\right)^{2}B^{-1}\ \ \text{and}\ \ \left\langle \Psi,\,\mathcal{N}\Psi\right\rangle \leq C_{M}\left(\ln B\right)^{2}.
\]
\end{cor}

\begin{proof}
The corollary follows from Lemma \ref{Lem: no commute}. 
\end{proof}
\begin{lem}
\label{lem:Last Blow}Let $\mathcal{K}>8\alpha/\pi$ and $1<\mathcal{A}<\sqrt{B/\ln B}$.
Denoting $\kappa=1-\left(8\alpha/\pi\mathcal{K}\right)$ for every
$M\in\mathbb{R}$ there exists a constant $C_{M}>0$ such that for
$B\ge C_{M}$ and all $\Psi\in H_{A}^{1}(\mathbb{R}^{3})\otimes\text{dom}\left(\sqrt{\mathcal{N}}\right)$
satisfying (\ref{eq:U})
\begin{align*}
\left\langle \Psi,\mathbb{H}(B)\Psi\right\rangle \geq & \left\langle P_{0}^{B}\Psi,\left(\mathfrak{h}_{\mathcal{K}}^{\text{co}}-\beta\left(1/\left(1-\mathcal{A}^{-1}\right)\right)\left|x\right|^{-1}\right)P_{0}^{B}\Psi\right\rangle +\kappa B\|P_{>}^{B}\Psi\|^{2}\\
 & -C_{M}\left(\ln B\right)^{2}\left(\mathcal{K}B^{-1}+\sqrt{\mathcal{K}B^{-1}}\right)-C_{M}\kappa^{-1}\ln B-1/4.
\end{align*}
\end{lem}

\begin{proof}
It can be argued as in the proof of Lemma \ref{Lem: no commute} with
$0<\eta<1$ 
\begin{align*}
 & \kappa\left\langle \Psi,\left(H_{B}-\partial_{3}^{2}\right)\Psi\right\rangle -\left\langle \Psi,\,\beta|x|^{-1}\Psi\right\rangle \\
 & \geq\kappa\left\langle P_{0}^{B}\Psi,\left(H_{B}-\partial_{3}^{2}\right)P_{0}^{B}\Psi\right\rangle -\left(\mathcal{A}/\left(\mathcal{A}-1\right)\right)\left\langle P_{0}^{B}\Psi,\beta|x|^{-1}P_{0}^{B}\Psi\right\rangle \\
 & \ \ +\kappa B\left\Vert P_{>}^{B}\Psi\right\Vert ^{2}+\left[\kappa\left(2B-3B\eta\right)-\mathcal{A}^{2}\beta^{2}\left(4\eta\kappa\right)^{-1}\right]\left\Vert P_{>}^{B}\Psi\right\Vert ^{2}.
\end{align*}
It now follows from Lemma \ref{Ultraviolet} that 
\begin{align*}
\left\langle \Psi,\mathbb{H}(B)\Psi\right\rangle \geq & \left\langle P_{0}^{B}\Psi,\left(\mathfrak{h}_{\mathcal{K}}^{\text{co}}-\beta\left(1/\left(1-\mathcal{A}^{-1}\right)\right)\left|x\right|^{-1}\right)P_{0}^{B}\Psi\right\rangle +\kappa B\|P_{>}^{B}\Psi\|^{2}\\
 & +\left[\kappa\left(2B-3B\eta\right)-\mathcal{A}^{2}\beta^{2}\left(4\eta\kappa\right)^{-1}\right]\left\Vert P_{>}^{B}\Psi\right\Vert ^{2}-\frac{1}{4}\\
 & +\left\langle P_{>}^{B}\Psi,\ \left(\int_{\Gamma_{\mathcal{K}}}a_{k}^{\dagger}a_{k}+\frac{\sqrt{\alpha}}{2\pi}\frac{a_{k}e^{ik\cdot x}}{|k|}+\frac{\sqrt{\alpha}}{2\pi}\frac{a_{k}^{\dagger}e^{-ik\cdot x}}{|k|}dk\right)P_{>}^{B}\Psi\right\rangle \\
 & +\left\langle P_{0}^{B}\Psi,\ \left(\frac{\sqrt{\alpha}}{2\pi}\int_{\Gamma_{\mathcal{K}}}\frac{a_{k}e^{ik\cdot x}}{|k|}+\frac{a_{k}^{\dagger}e^{-ik\cdot x}}{|k|}dk\right)P_{>}^{B}\Psi\right\rangle \\
 & +\left\langle P_{>}^{B}\Psi,\ \left(\frac{\sqrt{\alpha}}{2\pi}\int_{\Gamma_{\mathcal{K}}}\frac{a_{k}e^{ik\cdot x}}{|k|}+\frac{a_{k}^{\dagger}e^{-ik\cdot x}}{|k|}dk\right)P_{0}^{B}\Psi\right\rangle .
\end{align*}
By completion of square 
\begin{align*}
 & \int_{\Gamma_{\mathcal{K}}}\left(a_{k}^{\dagger}a_{k}+\frac{\sqrt{\alpha}}{2\pi}\frac{a_{k}}{|k|}e^{ik\cdot x}+\frac{\sqrt{\alpha}}{2\pi}\frac{a_{k}^{\dagger}}{|k|}e^{-ik\cdot x}\right)dk\\
 & \geq-\frac{\alpha}{4\pi^{2}}\int_{\Gamma_{\mathcal{K}}}\frac{1}{|k|^{2}}\,dk_{\perp}dk_{3}\\
 & =-\frac{\alpha\left(2\ln(2)+\pi\right)}{4\pi}\mathcal{K},
\end{align*}
and by Corollary \ref{Concentration} for $B\ge C_{M}$ 
\begin{align*}
 & \left\langle P_{>}^{B}\Psi,\ \left(\int_{\Gamma_{\mathcal{K}}}a_{k}^{\dagger}a_{k}+\frac{\sqrt{\alpha}}{2\pi}\frac{a_{k}e^{ik\cdot x}}{|k|}+\frac{\sqrt{\alpha}}{2\pi}\frac{a_{k}^{\dagger}e^{-ik\cdot x}}{|k|}dk\right)P_{>}^{B}\Psi\right\rangle \\
 & \geq-\frac{\alpha\left(2\ln(2)+\pi\right)}{4\pi}\mathcal{K}\|P_{>}^{B}\Psi\|^{2}\\
 & \geq-C_{M}\mathcal{K}B^{-1}\left(\ln B\right)^{2}.
\end{align*}

\noindent Furthermore it can be argued as in the proof of Lemma 5.4
in \cite{Frank-Geisinger} and using Corollary \ref{Concentration}
that for $B\geq C_{M}$
\begin{align*}
\left|\left\langle P_{0}^{B}\Psi,\ \left(\frac{\sqrt{\alpha}}{2\pi}\int_{\Gamma_{\mathcal{K}}}\frac{a_{k}}{|k|}e^{ik\cdot x}dk\right)P_{>}^{B}\Psi\right\rangle \right|
\end{align*}
\begin{align*}
 & \leq C\sqrt{\mathcal{K}}\|P_{>}^{B}\Psi\|\|\sqrt{\mathcal{N}+1}P_{0}^{B}\Psi\|\ \ \ \ \ \ \ \ \ \ \ \ \\
 & \leq C_{M}\sqrt{\mathcal{K}B^{-1}}\left(\ln B\right)^{2}.
\end{align*}
The remaining interaction terms are estimated similarly. Choosing
$\eta=2/3$ the lemma follows from Corollary \ref{Concentration}. 
\end{proof}
\begin{prop}
\label{prop: ending proposition} There exists a constant $C>0$ such
that for $B,\kappa$ and $\mathcal{K}$ satisfying $B\geq C$, $C\left(\ln B\right)^{-1/2}\leq\kappa\leq C^{-1}\ln B$,
$\mathcal{K}\geq\sqrt{B}$ and some $1<\mathcal{A}<\sqrt{B/\ln B}$
\begin{align*}
 & P_{0}^{B}\left(\mathfrak{h}_{\mathcal{K}}^{\text{co}}-\beta\left(1/\left(1-\mathcal{A}^{-1}\right)\right)\left|x\right|^{-1}\right)P_{0}^{B}\\
 & \geq\left(\kappa B+\kappa^{-1}\left(\ln B\right)^{2}\mathfrak{e}_{_{0}}-C\kappa^{-1/2}\left(\ln B\right)^{3/2}-C\left(1+\kappa^{-2}\right)\ln B\right)P_{0}^{B}.
\end{align*}
\end{prop}

\noindent The proof of Proposition \ref{prop: ending proposition}
shall be provided in Subsection \ref{subsec:Proof-of-Proposition}. 
\begin{proof}[Proof of Theorem \ref{thm:Lower Bound}.]
Fix $M>\mathfrak{e}_{_{0}}$. For $B$ large by Theorem \ref{Upper Bound}
there exist wave functions satisfying (\ref{eq:U}). It suffices to
argue the desired lower bound on $\left\langle \Psi,\mathbb{H}(B)\Psi\right\rangle $
with those wave functions. By Lemma \ref{lem:Last Blow} and Proposition
\ref{prop: ending proposition} 
\begin{align*}
 & \left\langle \Psi,\mathbb{H}(B)\Psi\right\rangle \\
 & \geq\kappa B+\left(\kappa^{-1}\left(\ln B\right)^{2}\mathfrak{e}_{_{0}}-C\kappa^{-1/2}\left(\ln B\right)^{3/2}-C\left(1+\kappa^{-2}\right)\ln B\right)\left\Vert P_{0}^{B}\Psi\right\Vert ^{2}\\
 & \ \ -C\left(\ln B\right)^{2}\left(\mathcal{K}B^{-1}+\sqrt{\mathcal{K}B^{-1}}\right)-C\kappa^{-1}\ln B-C
\end{align*}
with $\kappa=1-\left(8\alpha/\pi\mathcal{K}\right)$. Choosing $\mathcal{K}=B\left(\ln B\right)^{-4/3}$
and since $\left\Vert P_{0}^{B}\Psi\right\Vert \leq1$ 
\[
\left\langle \Psi,\mathbb{H}(B)\Psi\right\rangle \geq B+\mathfrak{e}_{_{0}}\left(\ln B\right)^{2}-C\left(\ln B\right)^{3/2}
\]
which is the claimed lower bound. 
\end{proof}

\subsection{\label{subsec:Proof-of-Proposition}Proof of Proposition \ref{prop: ending proposition}.}

\subsubsection{Reduction to one dimension}

In \cite{Frank-Geisinger} the authors consider a one-dimensional
Hamiltonian with $0<\mathcal{K}_{3}\leq\mathcal{K}$ and $1\leq\mathcal{K}_{\perp}\leq\mathcal{K}$
\begin{align*}
 & \mathfrak{h}^{1\text{d}}:=\\
 & \kappa_{1}\left(-\partial_{3}^{2}\right)+\int_{|k_{3}|\leq\mathcal{K}_{3}}\widehat{a}_{k_{3}}^{\dagger}\widehat{a}_{k_{3}}dk_{3}+\frac{\sqrt{\alpha}}{2\pi}\int_{|k_{3}|\leq\mathcal{K}_{3}}\nu\left(k_{3}\right)\left(\widehat{a}_{k_{3}}e^{ik_{3}x_{3}}+\widehat{a}_{k_{3}}^{\dagger}e^{-ik_{3}x_{3}}\right)dk_{3}
\end{align*}
acting on $L^{2}(\mathbb{R}^{3})\otimes\mathcal{F}$ with $\kappa_{1}:=\kappa-\left(8\alpha\slash\pi\mathcal{K}_{3}\right)\int_{0}^{\infty}\left(1+t\right)^{-1}\exp\left(-t\mathcal{K}_{3}^{2}\slash2B\right)\,dt$
and $\kappa$ as in the statement of Proposition \ref{prop: ending proposition},
and 
\[
\widehat{a}_{k_{3}}:=\frac{1}{\nu\left(k_{3}\right)}\int_{1\leq\left|k_{\perp}\right|\leq\mathcal{K}_{\perp}}\frac{a_{k}}{\left|k\right|}e^{ik_{\perp}\cdot x_{\perp}}dk_{\perp}
\]
with 
\begin{equation}
\nu(k_{3}):=\left(\int_{1\leq\left|k_{\perp}\right|\leq\mathcal{K}_{\perp}}\left|k\right|^{-2}dk_{\perp}\right)^{\frac{1}{2}}=\sqrt{\pi}\left(\ln\left(\mathcal{K}_{\perp}^{2}+k_{3}^{2}\right)-\ln\left(1+k_{3}^{2}\right)\right)^{\frac{1}{2}}\label{eq:Coupling Function}
\end{equation}
and satisfying $[\widehat{a}_{k_{3}},\widehat{a}_{k_{3}^{\prime}}^{\dagger}]=\delta\left(k_{3}-k_{3}^{\prime}\right)$
and $\left[\widehat{a}_{k_{3}},\widehat{a}_{k_{3}^{\prime}}\right]=[\widehat{a}_{k_{3}}^{\dagger},\widehat{a}_{k_{3}^{\prime}}^{\dagger}]=0$
for $k_{3},k_{3}^{\prime}\in\mathbb{R}$. 
\begin{lem}
\label{lem:Lemma 6.4 in FG}Denoting $\kappa_{2}=\kappa-2\alpha\pi^{-1}\mathcal{K}_{3}\mathcal{K}_{\perp}^{-2}$
\begin{align*}
 & P_{0}^{B}\left(\mathfrak{h}_{\mathcal{K}}^{\text{co}}-\beta\left(1/\left(1-\mathcal{A}^{-1}\right)\right)\left|x\right|^{-1}\right)P_{0}^{B}\\
 & \geq\kappa_{2}BP_{0}^{B}+P_{0}^{B}\left(\mathfrak{h}^{1\text{d}}-\beta\left(1/\left(1-\mathcal{A}^{-1}\right)\right)P_{0}^{B}\left|x\right|^{-1}P_{0}^{B}\right)P_{0}^{B}-\left(1+\frac{\alpha}{2}\right)P_{0}^{B}.
\end{align*}
\end{lem}

\begin{proof}
The lemma follows from Lemma 6.1, Lemma 6.2, Lemma 6.3 and Lemma 6.4
in \cite{Frank-Geisinger}. The proof of Lemma 6.3 in \cite{Frank-Geisinger}
uses an incorrect vector operator for cutting off high modes in the
$k_{\perp}$-direction cf. \cite{Lieb-Yamazaki}; the mistake can
be fixed, and the lemma stands true. 
\end{proof}

\subsubsection{Localization and decomposition.}

In \cite{Frank-Geisinger} the authors decompose the mode space into
$\mathcal{M}$ intervals, indexed with $b$, each of length $2\mathcal{K}_{3}/\mathcal{M}$
and consider for $u\in\mathbb{R}$ and $0<\gamma<1$ the block Hamiltonian
\[
h_{\gamma}^{(u)}:=\kappa_{1}\left(-\partial_{3}^{2}\right)+\sum_{b}\left[\left(1-\gamma\right)A_{b}^{(u)*}A_{b}^{(u)}+\frac{\sqrt{\alpha}}{2\pi}V(b)\left(A_{b}^{(u)}e^{ik_{b}x_{3}}+A_{b}^{(u)*}e^{-ik_{b}x_{3}}\right)\right]
\]
with $k_{b}$ a mode in block $b$ and the block creation and annihilation
operators $A_{b}^{(u)*}$ and $A_{b}^{(u)}$ acting on $\mathcal{F}\left(L^{2}\left(\mathbb{R}^{3}\right)\right)$,
where 
\[
A_{b}^{(u)}:=\frac{1}{V(b)}\int_{b}\nu\left(k_{3}\right)e^{i\left(k_{3}-k_{b}\right)u}\hat{a}_{k_{3}}dk_{3}
\]
with $V\left(b\right):=\left(\int_{b}\nu\left(k_{3}\right)^{2}dk_{3}\right)^{\frac{1}{2}}$.
Furthermore
\[
[A_{b}^{(u)},\,A_{b^{\prime}}^{(u)*}]=\delta_{bb^{\prime}},\ [A_{b}^{(u)},\,A_{b^{\prime}}^{(u)}]=[A_{b}^{(u)*},\,A_{b^{\prime}}^{(u)*}]=0\ \ \text{for all blocks}\ b,b^{\prime}.
\]

\begin{lem}
\label{lem:block estimate}For $\chi\in C_{0}^{\infty}\left(\mathbb{R}\right),\|\chi\|_{2}=1$
a nonnegative function supported on the interval $[-1/2,1/2]$ and
for $J>0$, denoting $\chi_{u}^{J}(x_{3})=J^{-1/2}\chi\left(J^{-1}\left(x_{3}-u\right)\right)$,

\begin{align*}
 & \mathfrak{h}^{1\text{d}}-\beta\left(1/\left(1-\mathcal{A}^{-1}\right)\right)P_{0}^{B}\left|x\right|^{-1}P_{0}^{B}\\
 & \geq\int_{\mathbb{R}}\chi_{u}^{J}\left[h_{\gamma}^{(u)}-\beta\left(1/\left(1-\mathcal{A}^{-1}\right)\right)P_{0}^{B}\left|x\right|^{-1}P_{0}^{B}\right]\chi_{u}^{J}\,du-\frac{\alpha\mathcal{K}_{3}^{2}J^{2}}{4\pi^{2}\gamma\mathcal{M}^{2}}R-\|\chi^{\prime}\|_{2}^{2}J^{-2}
\end{align*}
with 
\[
R:=\int_{|k_{3}|\leq\mathcal{K}_{3}}\nu\left(k_{3}\right)^{2}dk_{3}=\pi\int_{|k_{3}|\leq\mathcal{K}_{3}}\left(\ln\left(\mathcal{K}_{\perp}+k_{3}^{2}\right)-\ln\left(1+k_{3}^{2}\right)\right)dk_{3}.
\]
\end{lem}

\begin{proof}
The lemma follows from Lemma 6.5 in \cite{Frank-Geisinger}. 
\end{proof}

\subsubsection{Error estimates}

Similarly as in \cite{Frank-Geisinger}, \cite{Ghanta} and \cite{Lieb Thomas}
representing the block creation and annihilation operators by coherent
state integrals and completing the square it follows for a suitably
chosen $k_{b}$ that 
\begin{equation}
h_{\gamma}^{(u)}-\beta\left(1/\left(1-\mathcal{A}^{-1}\right)\right)P_{0}^{B}\left|x\right|^{-1}P_{0}^{B}\geq I-\mathcal{M},\label{eq:6.10 in =00003D00005BFG=00003D00005D}
\end{equation}
where 
\begin{align*}
I:=\inf_{\|\phi\|_{2}=1}\left[\kappa_{1}\|\partial_{3}\phi\|_{_{2}}^{2}-\frac{\alpha}{4\pi^{2}\left(1-\gamma\right)}\right. & \int_{\mathbb{R}}\nu\left(k_{3}\right)^{2}\left|\int_{\mathbb{R}^{3}}e^{ik_{3}x_{3}}\left|\phi\left(x_{\perp},x_{3}\right)\right|^{2}dx\right|^{2}dk_{3}\\
 & \left.-\frac{\beta}{1-\mathcal{A}^{-1}}\int_{\mathbb{R}^{3}}\left|x\right|^{-1}\left|\left(P_{0}^{B}\phi\right)\left(x_{\perp},x_{3}\right)\right|^{2}dx\right].
\end{align*}
Combining (\ref{eq:6.10 in =00003D00005BFG=00003D00005D}) with Lemma
\ref{lem:block estimate} 
\[
\mathfrak{h}^{1\text{d}}-\beta\left(1/\left(1-\mathcal{A}^{-1}\right)\right)P_{0}^{B}\left|x\right|^{-1}P_{0}^{B}\geq I-\mathcal{M}-\frac{\alpha\mathcal{K}_{3}^{2}J^{2}}{4\pi^{2}\gamma\mathcal{M}^{2}}R-\|\chi^{\prime}\|_{2}^{2}J^{-2}.
\]
Now it follows from Lemma \ref{lem:Lemma 6.4 in FG} that for some
constant $C>0$
\begin{align}
 & P_{0}^{B}\left(\mathfrak{h}_{\mathcal{K}}^{\text{co}}-\beta\left(1/\left(1-\mathcal{A}^{-1}\right)\right)\left|x\right|^{-1}\right)P_{0}^{B}\nonumber \\
 & \geq\kappa BP_{0}^{B}+IP_{0}^{B}-C\left(\frac{\mathcal{K}_{3}B}{\mathcal{K}_{\perp}^{2}}+\mathcal{M}+\frac{\mathcal{K}_{3}^{2}J^{2}}{\gamma\mathcal{M}^{2}}R+\frac{1}{J^{2}}\right)P_{0}^{B}.\label{standoff}
\end{align}

\begin{lem}
\label{Lemma: estimate on I}For any $L>0$ and $\epsilon>0$ and
with $\mathcal{D}(B,L)$ as given in (\ref{eq:D(B,L)}) assuming 
\begin{equation}
\frac{2\ln\mathcal{K}_{\perp}}{1-\gamma}=\frac{\mu(B)}{1-\mathcal{A}^{-1}}\ \ \text{and}\ \ \tilde{\kappa}_{1}:=\kappa_{1}-\frac{4\beta\epsilon\left|\mathcal{D}(B,L)\right|\ln\mathcal{K}_{\perp}}{\mu(B)\left(1-\gamma\right)}>0\label{assumptions}
\end{equation}
with $\mu(B):=\ln B-2\ln\ln B$, 
\[
I\geq\frac{4\left(\ln\mathcal{K}_{\perp}\right)^{2}\mathfrak{e}_{_{0}}}{\tilde{\kappa}_{1}\left(1-\gamma\right)^{2}}-\frac{2\beta\ln\mathcal{K}_{\perp}}{\mu(B)\left(1-\gamma\right)}\left(\frac{432L^{2}}{\left|\mathcal{D}(B,L)\right|^{3}\epsilon^{3}}+\frac{1}{L}+\frac{\left|\mathcal{D}(B,L)\right|}{\epsilon}\right).
\]
\end{lem}

\begin{proof}
For $\phi\in L^{2}\left(\mathbb{R}^{3}\right)$, $\|\phi\|_{2}=1$
and $\tilde{\phi}(x_{3}):=\left(\int_{\mathbb{R}^{2}}\left|\phi\left(x_{\perp},x_{3}\right)\right|^{2}dx_{\perp}\right)^{1/2}$,
\begin{align*}
\kappa_{1} & \|\partial_{3}\phi\|_{2}^{2}-\frac{\alpha}{4\pi^{2}\left(1-\gamma\right)}\int_{\mathbb{R}}\nu\left(k_{3}\right)^{2}\left|\int_{\mathbb{R}^{3}}e^{ik_{3}x_{3}}\left|\phi\left(x_{\perp},x_{3}\right)\right|^{2}dx\right|^{2}dk_{3}\\
 & -\frac{\beta}{1-\mathcal{A}^{-1}}\int_{\mathbb{R}^{3}}\frac{\left|\left(P_{0}^{B}\phi\right)\left(x_{\perp},x_{3}\right)\right|^{2}}{\left|x\right|}dx\\
\geq & \kappa_{1}\|\partial_{3}\phi\|_{2}^{2}-\frac{\alpha\ln\mathcal{K}_{\perp}}{1-\gamma}\int_{\mathbb{R}}\left|\int_{\mathbb{R}^{2}}\left|\phi\left(x_{\perp},x_{3}\right)\right|^{2}dx_{\perp}\right|^{2}dx_{3}\\
 & -\frac{\beta}{1-\mathcal{A}^{-1}}\int_{\mathbb{R}^{3}}\frac{\left|\left(P_{0}^{B}\phi\right)\left(x_{\perp},x_{3}\right)\right|^{2}}{\left|x\right|}dx\\
\geq & \left(\kappa_{1}-\frac{2\beta\epsilon\left|\mathcal{D}(B,L)\right|}{1-\mathcal{A}^{-1}}\right)\|\partial_{3}\tilde{\phi}\|_{2}^{2}-\frac{\alpha\ln\mathcal{K}_{\perp}}{1-\gamma}\|\tilde{\phi}\|_{4}^{4}-\frac{\beta\mu(B)}{1-\mathcal{A}^{-1}}\left(\tilde{\phi}(0)\right)^{2}\\
 & -\frac{\beta}{1-\mathcal{A}^{-1}}\left(\frac{432L^{2}}{\left|\mathcal{D}(B,L)\right|^{3}\epsilon^{3}}+\frac{1}{L}+\frac{\left|\mathcal{D}(B,L)\right|}{\epsilon}\right)\\
\geq & \frac{4\left(\ln\mathcal{K}_{\perp}\right)^{2}\mathfrak{e}_{_{0}}}{\tilde{\kappa}_{1}\left(1-\gamma\right)^{2}}-\frac{2\beta\ln\mathcal{K}_{\perp}}{\mu(B)\left(1-\gamma\right)}\left(\frac{432L^{2}}{\left|\mathcal{D}(B,L)\right|^{3}\epsilon^{3}}+\frac{1}{L}+\frac{\left|\mathcal{D}(B,L)\right|}{\epsilon}\right).
\end{align*}
At the first inequality it is used that 
\begin{equation}
\nu\left(k_{3}\right)^{2}=\pi\left(\ln\left(\mathcal{K}_{\perp}^{2}+k_{3}^{2}\right)-\ln\left(1+k_{3}^{2}\right)\right)\leq2\pi\ln\mathcal{K}_{\perp}\label{Bound on =00003D00005Cnu(k)}
\end{equation}
along with Plancherel's identity. At the second inequality the bathtub
principle in Lemma \ref{Lemma: Bathtub} and the argument in the proof
of Corollary \ref{Hydrogen} apply mutatis mutandis. At the third
inequality the assumptions in (\ref{assumptions}) are used. 
\end{proof}
\noindent From Lemma \ref{Lemma: estimate on I} and further assuming
\begin{equation}
\kappa-\tilde{\kappa}_{1}\leq\frac{\kappa}{2}\ \ \text{and}\ \ \gamma\leq\frac{1}{2},\label{eq:assumptions}
\end{equation}
it can be seen there is a constant $C>0$ such that 
\begin{align*}
I\geq & \ 4\kappa^{-1}\left(\ln\mathcal{K}_{\perp}\right)^{2}\mathfrak{e}_{_{0}}\\
 & -C\left(\frac{\left(\ln\mathcal{K}_{\perp}\right)^{2}}{\kappa}\left(\frac{\kappa-\tilde{\kappa}_{1}}{\kappa}+\gamma\right)+\frac{\ln\mathcal{K}_{\perp}}{\mu(B)}\left(\frac{L^{2}}{\left|\mathcal{D}\left(B,L\right)\right|^{3}\epsilon^{3}}+\frac{1}{L}+\frac{\left|\mathcal{D}\left(B,L\right)\right|}{\epsilon}\right)\right).
\end{align*}
With the above bound and (\ref{standoff}) the argument in \cite{Frank-Geisinger}
applies mutatis mutandis and choosing $J^{2}=\kappa^{1/5}\mathcal{K}_{3}^{-3/5}\left(\ln\mathcal{K}_{\perp}\right)^{-3/5}$,
$\mathcal{M}=[J^{-2}]$ and $\gamma=\kappa^{4/5}\mathcal{K}_{3}^{3/5}\left(\ln\mathcal{K}_{\perp}\right)^{-7/5}$
yields
\begin{align*}
P_{0}^{B} & \left(\mathfrak{h}_{\mathcal{K}}^{\text{co}}-\beta\left(1/\left(1-\mathcal{A}^{-1}\right)\right)\left|x\right|^{-1}\right)P_{0}^{B}\\
\geq & \left(\kappa B+\frac{4\left(\ln\mathcal{K}_{\perp}\right)^{2}}{\kappa}\mathfrak{e}_{_{0}}\right)P_{0}^{B}\\
 & -C\left(\kappa^{-1}\left(\kappa-\kappa_{1}\right)\left(\ln\mathcal{K}_{\perp}\right)^{2}+\kappa^{-1/5}\mathcal{K}_{3}^{3/5}\left(\ln\mathcal{K}_{\perp}\right)^{3/5}+B\mathcal{K}_{3}\mathcal{K}_{\perp}^{-2}\right)P_{0}^{B}\\
 & -C\frac{\ln\mathcal{K}_{\perp}}{\mu(B)}\left(\frac{L^{2}\epsilon^{-3}}{\left|\mathcal{D}\left(B,L\right)\right|^{3}}+\frac{1}{L}+\frac{\left|\mathcal{D}\left(B,L\right)\right|}{\epsilon}+\frac{\epsilon\left|\mathcal{D}\left(B,L\right)\right|\left(\ln\mathcal{K}_{\perp}\right)^{2}}{\kappa^{2}}\right)P_{0}^{B}.
\end{align*}
It is shown in \cite{Frank-Geisinger} choosing $\mathcal{K}_{\perp}=B^{1/2}$
and $\mathcal{K}_{3}=\kappa^{-1/2}\left(\ln B\right)^{3/2}$ that
\[
\kappa^{-1}\left(\kappa-\kappa_{1}\right)\left(\ln\mathcal{K}_{\perp}\right)^{2}+\kappa^{-1/5}\mathcal{K}_{3}^{3/5}\left(\ln\mathcal{K}_{\perp}\right)^{3/5}+B\mathcal{K}_{3}\mathcal{K}_{\perp}^{-2}\leq C\kappa^{-1/2}\left(\ln B\right)^{3/2}.
\]
Now choosing $L=1/\ln B$ and $\epsilon=1/\ln B$, since $1\leq\left|\mathcal{D}\left(B,L\right)\right|\leq C$
for $B\geq C$, 
\[
\frac{\ln\mathcal{K}_{\perp}}{\mu(B)}\left(\frac{L^{2}\epsilon^{-3}}{\left|\mathcal{D}\left(B,L\right)\right|^{3}}+\frac{1}{L}+\frac{\left|\mathcal{D}\left(B,L\right)\right|}{\epsilon}+\frac{\epsilon\left|\mathcal{D}\left(B,L\right)\right|\left(\ln\mathcal{K}_{\perp}\right)^{2}}{\kappa^{2}}\right)\leq C\left(1+\kappa^{-2}\right)\ln B.
\]
With the above choice of parameters and the assumptions $C\left(\ln B\right)^{-1/2}\leq\kappa\leq C^{-1}\ln B$
and $\mathcal{K}\geq\sqrt{B}$, the conditions in (\ref{assumptions})
and (\ref{eq:assumptions}) and that $0<\mathcal{K}_{3}\leq\mathcal{K},$
$1\leq\mathcal{K}_{\perp}\leq\mathcal{K}$ and $1<\mathcal{A}<\sqrt{B/\ln B}$
are verified, thereby concluding the proof
\begin{align*}
 & P_{0}^{B}\left(\mathfrak{h}_{\mathcal{K}}^{\text{co}}-\beta\left(1/\left(1-\mathcal{A}^{-1}\right)\right)\left|x\right|^{-1}\right)P_{0}^{B}\\
 & \geq\left(\kappa B+\kappa^{-1}\left(\ln B\right)^{2}\mathfrak{e}_{_{0}}-C\kappa^{-1/2}\left(\ln B\right)^{3/2}-C\left(1+\kappa^{-2}\right)\ln B\right)P_{0}^{B}.
\end{align*}

\begin{cor}
\label{cor:Lower Bound}Let $W$ be a sum of a bounded Borel measure
on the real line and a $L^{\infty}(\mathbb{R})$ function. With $\epsilon$
a real parameter and $E_{\epsilon}(B)$ and $\mathfrak{e}_{_{\epsilon}}$
as in Corollary \ref{cor:Upper} there is a constant $C>0$ such that
for $B\geq C$ 
\[
E_{\epsilon}(B)\geq B+\mathfrak{e}_{_{\epsilon}}\left(\ln B\right)^{2}-C\left(\ln B\right)^{3/2}.
\]
\end{cor}

\begin{proof}
The above arguments apply mutatis mutandis; see also \cite{Griesemer-Wellig}.
\end{proof}

\section{\label{sec:Proof-of-Theorem 2.2}Proof of Theorem \ref{thm:Wavefunction}}
\begin{proof}[Proof of Theorem \ref{thm:Wavefunction}.]
Let $E_{\epsilon}(B)$ be the ground-state energy of the Hamiltonian
$\mathbb{H}_{\epsilon}(B)$ in (\ref{eq:pert}) and let the one-dimensional
energy $\mathfrak{e}_{_{\epsilon}}$ be as defined in (\ref{eq:pert2}).
The large $B$ asymptotics of $E_{\epsilon}(B)$ in (\ref{eq:1-3})
follows from Corollary \ref{cor:Upper} and Corollary \ref{cor:Lower Bound}.
As explained in the introduction it follows from the variational principle,
Theorem \ref{thm:Energy Asymptotics} and (\ref{eq:1-3}) that for
$\epsilon>0$ 
\[
\frac{\mathfrak{e}_{_{0}}-\mathfrak{e}_{_{\epsilon}}}{\epsilon}\geq\limsup_{B\rightarrow\infty}\frac{1}{\left(\ln B\right)}\int_{\mathbb{R}}W\left(x_{3}\right)\left(\int_{\mathbb{R}^{2}}\left\Vert \Psi^{(B)}\right\Vert _{\mathcal{F}}^{2}\left(x_{\perp},\frac{x_{3}}{\ln B}\right)dx_{\perp}\right)dx_{3}
\]
and for $\epsilon<0$ 
\[
\frac{\mathfrak{e}_{_{0}}-\mathfrak{e}_{_{\epsilon}}}{\epsilon}\leq\liminf_{B\rightarrow\infty}\frac{1}{\left(\ln B\right)}\int_{\mathbb{R}}W\left(x_{3}\right)\left(\int_{\mathbb{R}^{2}}\left\Vert \Psi^{(B)}\right\Vert _{\mathcal{F}}^{2}\left(x_{\perp},\frac{x_{3}}{\ln B}\right)dx_{\perp}\right)dx_{3}.
\]
Theorem \ref{thm:Wavefunction} now follows from Theorem \ref{thm:Differentiation}. 
\end{proof}

\subsubsection*{Acknowledgement}

\noindent I thank Michael Loss for several valuable suggestions and
Rupert Frank for the general strategy used in the proof of Theorem
\ref{thm:Wavefunction}.

\appendix

\section{Compactness of Minimizing Sequences}
\begin{thm}
If a sequence $\left\{ \psi_{_{n}}\right\} _{n=1}^{\infty},\left\Vert \psi_{_{n}}\right\Vert _{2}=1$
satisfies $\underset{n\rightarrow\infty}{\lim}\mathcal{E}_{_{0}}\left(\psi_{_{n}}\right)=\mathfrak{e}_{_{0}}$
with the functional $\mathcal{E}_{_{0}}$ as given in (\ref{eq: Pekar Functional}),
then there exists a subsequence $\{\psi_{_{n_{k}}}\}_{k=1}^{\infty}$
and some $\psi\in H^{1}\left(\mathbb{R}\right)$ such that $\left\Vert \psi\right\Vert _{2}=1$,
$\mathcal{E}_{_{0}}\left(\psi\right)=\mathfrak{e}_{_{0}}$ and $\|\psi_{_{n_{k}}}-\psi\|_{H^{1}}\rightarrow0$
as $k\rightarrow\infty$. 
\end{thm}

\begin{proof}
For $\varphi\in H^{1}\left(\mathbb{R}\right)$ it can be argued as
in the proof of Theorem \ref{thm:Differentiation} that 
\[
\mathcal{E}_{_{0}}\left(\varphi\right)\geq\frac{3}{4}\|\varphi^{\prime}\|_{2}^{2}-\|\varphi\|_{2}^{2}\left(\frac{\alpha}{2}\|\varphi\|_{2}^{2}+\beta\right)^{2}.
\]
Furthermore since $\psi_{_{n}}$ is a minimizing sequence, $\mathcal{E}_{_{0}}\left(\psi_{_{n}}\right)<\mathfrak{e}_{_{0}}+1$
for $n$ large. Then $\left\Vert \psi_{_{n}}\right\Vert _{H^{1}}<C$
and there exists a subsequence $\{\psi_{_{n_{k}}}\}_{k=1}^{\infty}$
and some $\psi\in H^{1}\left(\mathbb{R}\right)$ such that $\psi_{n_{k}}$
converges to $\psi$ weakly in $H^{1}\left(\mathbb{R}\right)$.

\noindent \textbf{Step 1} (Compactness). It shall be argued that the
subsequence $\{\psi_{n_{k}}\}_{k=1}^{\infty}$ satisfies 
\begin{equation}
\forall\delta>0,\ \exists R<\infty\ \ \text{s.t.}\ \ \|\psi_{n_{k}}\|_{L^{2}\left(\left\{ \left|x\right|<R\right\} \right)}^{2}>1-\delta.\label{eq:C}
\end{equation}
Essential to the argument is the binding inequality $\mathfrak{e}_{_{0}}<\mathfrak{e}_{_{T}}$
where $\mathfrak{e}_{_{T}}:=\underset{\left\Vert \varphi\right\Vert _{2}=1}{\inf}\mathcal{E}_{_{T}}\left(\varphi\right)$
and $\mathcal{E}_{_{T}}\left(\varphi\right):=\int_{\mathbb{R}}\left|\varphi'\right|^{2}dx-(\alpha/2)\int_{\mathbb{R}}\left|\varphi\right|^{4}dx$
is the translation-invariant problem admitting a symmetric decreasing
minimizer $\phi_{_{T}}$ \cite{Frank-Geisinger}. Indeed 
\[
\mathfrak{e}_{_{0}}\leq\mathcal{E}_{_{0}}\left(\phi_{_{T}}\right)=\mathcal{E}_{_{T}}\left(\phi_{_{T}}\right)-\beta\phi_{_{T}}^{2}(0)=\mathfrak{e}_{_{T}}-\beta\phi_{_{T}}^{2}(0)<\mathfrak{e}_{_{T}}.
\]
Moreover it should be noted 
\begin{equation}
\mathcal{E}_{_{0}}\left(\varphi\right)\geq\mathfrak{e}_{_{0}}\left\Vert \varphi\right\Vert _{2}^{2}\ \ \ \text{and}\ \ \ \mathcal{E}_{_{T}}\left(\varphi\right)\geq\mathfrak{e}_{_{T}}\left\Vert \varphi\right\Vert _{2}^{2}\ \ \text{when}\ \ \left\Vert \varphi\right\Vert _{2}\leq1.\label{eq:C.21}
\end{equation}
Also a quadratic partition of unity is chosen, $\chi^{2}+\tilde{\chi}^{2}\equiv1$,
where $0\leq\chi\leq1$ is a smooth function with $\chi(x)=1$ when
$|x|<1/2$ and $\chi(x)=0$ when $|x|>1$. Denoting $\chi_{_{_{R}}}=\chi\left(R^{-1}\cdot\right)$
it follows from (\ref{eq:C.21}) that 
\begin{align}
\mathcal{E}_{_{0}}\left(\psi_{n_{k}}\right)= & \,\mathcal{E}_{_{0}}\left(\chi_{_{R}}\psi_{n_{k}}\right)+\mathcal{E}_{_{T}}\left(\tilde{\chi}_{_{R}}\psi_{n_{k}}\right)\nonumber \\
 & -2\int_{\mathbb{R}}\chi_{_{R}}^{2}\tilde{\chi}_{_{R}}^{2}\left|\psi_{n_{k}}\right|^{4}dx-\int_{\mathbb{R}}\left|\psi_{n_{k}}\right|^{2}\left(\left|\chi_{_{R}}^{\prime}\right|^{2}+\left|\tilde{\chi}_{_{R}}^{\prime}\right|^{2}\right)dx\nonumber \\
\geq & \left(\mathfrak{e}_{_{0}}-\mathfrak{e}_{_{T}}\right)\left\Vert \chi_{_{R}}\psi_{n_{k}}\right\Vert _{2}^{2}+\mathfrak{e}_{_{T}}\nonumber \\
 & -2\int_{\mathbb{R}}\chi_{_{R}}^{2}\tilde{\chi}_{_{R}}^{2}\left|\psi_{n_{k}}\right|^{4}dx-\int_{\mathbb{R}}\left|\psi_{n_{k}}\right|^{2}\left(\left|\chi_{_{R}}^{\prime}\right|^{2}+\left|\tilde{\chi}_{_{R}}^{\prime}\right|^{2}\right)dx.\label{eq:two terms}
\end{align}
Since $\chi,\tilde{\chi}$ have bounded derivatives 
\[
\int_{\mathbb{R}}\left|\psi_{_{n_{k}}}\right|^{2}\left(\left|\chi_{_{R}}^{\prime}\right|^{2}+\left|\tilde{\chi}_{_{R}}^{\prime}\right|^{2}\right)dx<CR^{-2}
\]
for some $C>0$. Furthermore with $D_{R}:=\left\{ R/2\leq|x|\leq R\right\} $
\[
\int_{\mathbb{R}}\chi_{_{R}}^{2}\tilde{\chi}_{_{R}}^{2}\left|\psi_{n_{k}}\right|^{4}dx\leq\|\psi_{_{n_{k}}}\|_{L^{^{4}}\left(D_{R}\right)}^{4},\ \text{and}\ \|\psi_{_{n_{k}}}\|_{L^{^{4}}\left(D_{R}\right)}^{4}\longrightarrow\|\psi\|_{L^{^{4}}\left(D_{R}\right)}^{4}
\]
by Rellich-Kondrashov, so the first term in (\ref{eq:two terms})
can also be made arbitrarily small with $R$ chosen to be large enough
uniformly in $k$. Hence for any $\delta>0$ there is some $R$ such
that for all $k$ 
\begin{equation}
\mathcal{E}_{_{0}}\left(\psi_{n_{k}}\right)\geq\left(\mathfrak{e}_{_{0}}-\mathfrak{e}_{_{T}}\right)\left\Vert \chi_{_{R}}\psi_{n_{k}}\right\Vert _{2}^{2}+\mathfrak{e}_{_{T}}-\delta(\mathfrak{e}_{_{T}}-\mathfrak{e}_{_{0}})/2.\label{eq:oneside}
\end{equation}
Since $\{\psi_{n_{k}}\}_{k=1}^{\infty}$ is a minimizing sequence
for $\mathfrak{e}_{_{0}}$, for $k$ large 
\begin{equation}
\mathcal{E}_{_{0}}\left(\psi_{n_{k}}\right)\leq\mathfrak{e}_{_{0}}+\delta(\mathfrak{e}_{_{T}}-\mathfrak{e}_{_{0}})/2.\label{eq:otherside}
\end{equation}
Compactness now follows from (\ref{eq:oneside}) and (\ref{eq:otherside}).

\noindent \textbf{Step 2} (Weak Limit is a Minimizer). By Rellich-Kondrashov
and the compactness property in (\ref{eq:C}) 
\begin{equation}
\|\psi_{n_{k}}-\psi\|_{2}\rightarrow0\ \ \ \text{and}\ \ \ \|\psi\|_{2}=1.\label{eq:A.II.1}
\end{equation}
Since $\left\Vert \psi_{n}\right\Vert _{H^{1}}<C$, by Sobolev and
Hölder's inequalities 
\[
\int_{\mathbb{R}}\left(|\psi_{n_{k}}|^{4}-\left|\psi\right|^{4}\right)dx\leq C\|\psi_{n_{k}}-\psi\|_{2}\longrightarrow0.
\]
Furthermore by Theorem 8.6 in \cite{Lieb-Loss} $\psi_{n_{k}}(0)\rightarrow\psi\left(0\right)$,
so 
\begin{equation}
\frac{\alpha}{2}\int_{\mathbb{R}}|\psi_{n_{k}}|^{4}dx+\beta|\psi_{n_{k}}\left(0\right)|^{2}\rightarrow\frac{\alpha}{2}\int_{\mathbb{R}}\left|\psi\right|^{4}dx+\beta\left|\psi\left(0\right)\right|^{2}.\label{eq:A.II.3}
\end{equation}
Then since $\underset{k\rightarrow\infty}{\liminf}\,\|\psi_{n_{k}}^{\prime}\|_{2}\geq\|\psi^{\prime}\|_{2}$,
$\mathfrak{e}_{_{0}}=\underset{k\rightarrow\infty}{\lim}\mathcal{E}_{_{0}}(\psi_{n_{k}})\geq\mathcal{E}_{_{0}}\left(\psi\right)\geq\mathfrak{e}_{_{0}}$
and $\mathcal{E}_{_{0}}\left(\psi\right)=\mathfrak{e}_{_{0}}$.

\noindent \textbf{Step 3} (Convergence in $H^{1}(\mathbb{R})$). From
(\ref{eq:A.II.3}) 
\begin{align*}
\lim_{k\rightarrow\infty}\|\psi_{n_{k}}^{\prime}\|_{2}^{2} & =\lim_{k\rightarrow\infty}\left(\mathcal{E}_{_{0}}(\psi_{n_{k}})+\frac{\alpha}{2}\|\psi_{n_{k}}\|_{4}^{4}+\beta\left|\psi_{n_{k}}(0)\right|^{2}\right)\\
 & =\mathfrak{e}_{_{0}}+\frac{\alpha}{2}\|\psi\|_{4}^{4}+\beta\left|\psi\left(0\right)\right|^{2}=\|\psi^{\prime}\|_{2}^{2},
\end{align*}
and since $\psi_{_{n_{k}}}\rightharpoonup\psi$ in $H^{1}$, $\|\psi_{n_{k}}^{\prime}-\psi^{\prime}\|_{2}\rightarrow0.$
Strong convergence in $H^{1}$ now follows from (\ref{eq:A.II.1}). 
\end{proof}

\section{Bound on the Effective Coulomb Potential}

\noindent Recalling the effective Coulomb potential $V_{\mathcal{U}}^{B}$
from (\ref{eq:uppercut}), 
\begin{lem}
\label{lem:B-label1}For any $L>0$ and $\phi\in H^{1}(\mathbb{R})$
one has for $B>1$ 
\begin{align*}
 & \left|\int_{\mathbb{R}}V_{\mathcal{U}}^{B}\left(x\right)\left|\phi\left(x\right)\right|^{2}dx-\left(\ln B-2\ln\ln B\right)\left|\phi(0)\right|^{2}\right|\\
 & \leq L^{-1}\|\phi\|_{2}^{2}+8\sqrt{L}\|\phi^{\prime}\|_{2}^{3/2}\|\phi\|_{2}^{1/2}+\left|\mathcal{G}\left(B,L\right)\right|\|\phi^{\prime}\|_{2}\|\phi\|_{2};
\end{align*}
\[
\mathcal{G}\left(B,L\right):=2\ln L+2\ln\ln B+2\int_{0}^{\infty}e^{-u}\ln\left(\sqrt{\frac{1}{u}+\frac{2}{BL^{2}}}+\sqrt{\frac{1}{u}}\right)du-\ln2.
\]
\end{lem}

\begin{proof}
Writing $\int_{\mathbb{R}}V_{\mathcal{U}}^{B}\left(x\right)\left|\phi\left(x\right)\right|^{2}$
as
\[
|\phi(0)|^{2}\int_{|x|\le L}V_{\mathcal{U}}^{B}(x)+\int_{|x|\leq L}V_{\mathcal{U}}^{B}(x)\left(|\phi(x)|^{2}-|\phi(0)|^{2}\right)+\int_{|x|\geq L}V_{\mathcal{U}}^{B}(x)|\phi(x)|^{2}
\]
it is possible to bound 
\begin{align}
 & \int_{|x|\geq L}V_{\mathcal{U}}^{B}\left(x\right)\left|\phi(x)\right|^{2}dx\leq L^{-1}\int_{|x|\geq L}|\phi(x)|^{2}dx,\label{eq:c1-1}\\
 & \left|\int_{|x|\leq L}V_{\mathcal{U}}^{B}\left(x\right)\left(\left|\phi(x)\right|^{2}-\left|\phi(0)\right|^{2}\right)dx\right|\leq8\sqrt{L}\|\phi^{\prime}\|_{2}^{3/2}\|\phi\|_{2}^{1/2}\label{eq:c2-1}
\end{align}
and to evaluate the integral 
\[
\int_{|x|\le L}V_{\mathcal{U}}^{B}\left(x\right)dx=\ln B-2\ln\ln B+\mathcal{G}\left(B,L\right).
\]
To arrive at the bound in (\ref{eq:c2-1}) the following inequalities
are used 
\begin{equation}
\left|\phi(x)-\phi(0)\right|\leq\sqrt{\left|x\right|}\|\phi^{\prime}\|_{2}\ \text{and}\ \|\phi\|_{\infty}^{2}\leq\|\phi^{\prime}\|_{2}\|\phi\|_{2}.\label{eq:c3-1}
\end{equation}
The lemma now follows from (\ref{eq:c1-1}), (\ref{eq:c2-1}) and
the rightmost inequality of (\ref{eq:c3-1}). 
\end{proof}
\begin{cor}
\label{cor:B-2}For any $L>0$ and $\phi\in H^{1}(\mathbb{R})$ one
has for $B>1$ 
\begin{align*}
 & \left|\int\int_{\mathbb{R}\times\mathbb{R}}\left|\phi(x)\right|^{2}\frac{1}{\sqrt{2}}V_{\mathcal{U}}^{B}\left(\frac{x-y}{\sqrt{2}}\right)\left|\phi\left(y\right)\right|^{2}dxdy-\left(\ln B-2\ln\ln B\right)\|\phi\|_{4}^{4}\right|\\
 & \leq L^{-1}\|\phi\|_{2}^{4}+8\sqrt{L}\|\phi^{\prime}\|_{2}^{3/2}\|\phi\|_{2}^{5/2}+\left|\mathcal{G}\left(B,L/\sqrt{2}\right)\right|\|\phi^{\prime}\|_{2}\|\phi\|_{2}^{3}
\end{align*}
with $\mathcal{G}\left(B,L\right)$ as above. 
\end{cor}

\begin{proof}
The corollary follows from Lemma \ref{lem:B-label1}. 
\end{proof}
\noindent Now recalling the potential $V_{\mathcal{\mathcal{L}}}^{B}$
from (\ref{eq:Bath}), 
\begin{cor}
For any $L>0$ and $\phi\in H^{1}(\mathbb{R})$ one has for B>1 
\begin{align*}
 & \left|\int_{\mathbb{R}}V_{\mathcal{\mathcal{L}}}^{B}(x)\left|\phi\left(x\right)\right|^{2}dx-\left(\ln B-2\ln\ln B\right)\left|\phi\left(0\right)\right|^{2}\right|\\
 & \leq L^{-1}\|\phi\|_{2}^{2}+8\sqrt{L}\|\phi^{\prime}\|_{2}^{3/2}\|\phi\|_{2}^{1/2}+\left|\mathcal{D}(B,L)\right|\|\phi^{\prime}\|_{2}\|\phi\|_{2};
\end{align*}
\[
\mathcal{D}\left(B,L\right):=2\ln L+2\ln\ln B+\frac{2}{\sqrt{\frac{2}{BL^{2}}+1}+1}+2\ln\left(\sqrt{1+\frac{2}{BL^{2}}}+1\right)-\ln2.
\]
\end{cor}

\begin{proof}
Evaluating the integral
\[
\int_{\left|x\right|<L}V_{\mathcal{\mathcal{L}}}^{B}(x)\,dx=\ln B-2\ln\ln B+\mathcal{D}\left(B,L\right)
\]
 the argument follows the proof of Lemma \ref{lem:B-label1}. 
\end{proof}

\end{document}